%% file: conference_101719.tex
\documentclass[journal]{IEEEtran}
\usepackage[numbers,compress]{natbib}
\usepackage{amsfonts,mathrsfs,mathdots,mathtools,amsthm,amssymb,centernot}
\usepackage{graphicx}
\usepackage{textcomp}
\usepackage{xcolor}
\usepackage{tikz}
\usepackage{pgfplots}
\usepackage{subcaption}
\usepackage{dsfont}
\usepackage{pifont}
\usepackage{siunitx}
\usepackage{comment}
\usepackage{pgfplotstable}
\usepackage{authblk}
\usepackage[ruled,algo2e,vlined]{algorithm2e}
\usepackage{setspace}
\usepackage{hyperref}
\usepackage{color,soul}
\usepackage{bm}	
\usepackage[inline]{enumitem}
\usepackage{nicefrac}
\usepackage{dirtytalk}

\input{macros}
\begin{document}

\title{Extremal Mechanisms for Pointwise Maximal Leakage
}

\author{Leonhard~Grosse,~\IEEEmembership{Student~Member,~IEEE,}
    Sara~Saeidian,~\IEEEmembership{Member,~IEEE,}
    Tobias~J.~Oechtering,~\IEEEmembership{Senior~Member,~IEEE}\thanks{\modified{The work has been supported by the Swedish Research Council (VR) under grant 2023-04787 and  project DataLEASH of the Digital Future center funded by the Swedish government. Leonhard Grosse, Sara Saeidian and Tobias J. Oechtering are with the Division of Information Science and Engineering, School of Electrical Engineering and Computer Science, KTH Royal Institute of Technology, 100 44 Stockholm, Sweden, email: \{lgrosse, saeidian, oech\}@kth.se}.}}


\maketitle
\vspace{-2cm}
\begin{abstract}
Data publishing under privacy constraints can be achieved with mechanisms that add randomness to data points when released to an untrusted party, thereby decreasing the data's utility. In this paper, we analyze this privacy-utility tradeoff for the pointwise maximal leakage \modified{(PML)} privacy measure and \modified{provide optimal privacy mechanisms for} a general class of convex utility functions. \modified{PML} was recently proposed as an operationally meaningful privacy measure based on two equivalent threat models: An adversary guessing a randomized function and an adversary aiming to maximize a general gain function. \modified{We prove a cardinality bound, showing that output alphabets of optimal mechanisms in this context need not to be larger than the size of their inputs. Then, we characterize the optimization region as a (convex) polytope.} We derive closed-form optimal privacy \modified{mechanisms for arbitrary priors in the high privacy regime (when the privacy parameter is sufficiently small) and uniform priors for all ranges of the privacy parameter} using tools from convex analysis. Furthermore, we present a linear program that can compute optimal mechanisms for PML in a general setting. \modified{We conclude by demonstrating the performance of the closed-form mechanisms through numerical simulations.}
\end{abstract}


\section{Introduction}
\input{sections/introduction}

\section{Preliminaries}
\input{sections/preliminaries}

\section{The Privacy-utility Tradeoff Problem}

\input{sections/tradeoff_problem}

\section{Optimal Privacy Mechanisms}
\input{sections/opt_sol}

\subsection{Optimal Binary Mechanism}
\input{sections/binary}

\subsection{PML Privacy Regions}
\input{sections/regions}

\subsection{Optimal Mechanism in the High-Privacy Regime}
\input{sections/high_privacy}

\subsection{Optimal Mechanisms for Uniform Priors}
\input{sections/closedformunifgeneral}

\subsection{General Optimal Mechanisms via a Linear Program}
\input{sections/lp}

\modified{\subsection{Numerical Results}}
\input{sections/numericalResults}

\section{Conclusions}
\input{sections/conclusions}


%
\appendix
\input{sections/appendix}

\bibliographystyle{IEEEtranN}
\footnotesize
\bibliography{bibtex/bib/IEEEabrv, bibtex/bib/thisbib}

\end{document}

%% file: macros.tex
\DeclareMathOperator*{\indep}{\perp \!\!\! \perp}

\def\supp{\textnormal{supp}}
\def\epspml{\epsilon}
\def\epsldp{\epsilon_r}
\def\e{e}
\def\epsilon{\varepsilon}
\theoremstyle{definition}
\newtheorem{definition}{Definition}
\newtheorem{lemma}{Lemma}
\newtheorem{theorem}{Theorem}

\newtheorem{corollary}{Corollary}
\newtheorem{proposition}{Proposition}
\newtheorem{example}{Example}
\newtheorem{remark}{Remark}

\allowdisplaybreaks
\mathtoolsset{showonlyrefs}

\newcommand{\modified}[1]{{\color{black} #1}}

%% file: sections/introduction.tex
As policymakers are tasked with writing legislation to limit the negative influence of actors that are using individuals' personal data, the concept of provable privacy guarantees has moved into focus as a tool for better design and easier policing of electronic data processing systems \cite{nissim2018philosophicalprivacy,nissim2021codesignlawandprivacy}. To do this, numerous privacy measures have been proposed across different domains, each with its own strengths and limitations. With implementations in systems by Google \cite{googleRAPPOR} and Apple \cite{appleDP}, among others, \emph{differential privacy} (DP) \cite{dwork2014algorithmic} and its local variant, \emph{local differential privacy} (LDP) \cite{evfimievski2003limiting,kasiviswanathan2011can} are today often used in practical implementations. The privacy guarantee of differential privacy hinges on hiding participation: The outcome of any differentially private data release does not change significantly whether or not a specific individual's data is included in the analysis. This approach has previously been argued to define privacy as a \emph{causal} property of the processing algorithms \cite{tschantz2020sok}. While this interpretation conceptually poses a strong notion of privacy, it has been pointed out that in modern data processing systems, an \emph{associative} view of privacy would be desirable \cite{ghosh2016inferential}. Several works argue that such a guarantee from differential privacy requires independence assumptions on the database entries \cite{nofreelunch,yang2015bayesian,zhu2014correlated}. Another critique on differential privacy concerns its parameter and the parameter's relation to the provided privacy guarantee: In practice, differential privacy does not provide any clear guideline for how to pick the privacy level in order to achieve the desired privacy protection. In fact, a recent survey among system designers by \citet{dwork2019exposeyourepsilons} shows that the privacy parameter in real implementations is often picked arbitrarily. Works like \cite{ghosh2016inferential,nofreelunch} therefore argue in favor of adopting inferential guarantees, that is, guarantees that ensure that an adversaries knowledge does not change significantly from her prior knowledge upon observing the outcome of a mechanism.

Parallel to the works on differential privacy, a wide array of privacy measures have been proposed in the information theory literature. \modified{Many of these measures put forward a notion of \emph{information leakage}, quantified by various statistical quantities. The earliest example of this is mutual information \cite{asoodeh2015miprivacy,wang2016miprivacy}. While mutual information has a central role in communication theory, \citet{smith2009foundations} argues that the value of mutual information can be counter-intuitive in certain privacy problems. Other works like \cite{liao2019tunable} discuss generalizations of mutual information due to \citet{arimoto1977information} and \citet{sibson1969information} as privacy measures. These measures arise naturally when assessing privacy risks in specific threat models. Another line of work aims for a definition of privacy more in line with LDP, called \emph{local information privacy} (LIP) \cite{jiang2020LIP}. LIP imposes a symmetric upper and lower bound on the \emph{information density} between the secret and the released random variable. Later, LIP was generalized to allow for asymmetric bounds on the information density in \cite{zarrabian2023lift}}. Other approaches use the probability of correctly guessing \cite{asoodeh2018estimation} and various $f$-divergences \cite{rassouli2019TVDprivacy,8825803} as privacy measures. For a detailed survey of privacy measures, see \citet{bloch2021inftheoprivacyoverview} and \citet{wagner2019privacyoverview}. 

\modified{Among the notions of information leakage mentioned above}, \emph{operational} privacy measures pose promising alternatives to the de facto standard of (local) differential privacy. Operational measures of information leakage provide definitions of privacy building on concrete statistical threat models. These threat models have the advantage of making the type of privacy provided by a measure directly explainable to stakeholders. Further, since any assumptions made in the privacy guarantee are explicit in the model, operational measures avoid confusion about what type of privacy is or is not promised. One such operationally meaningful notion is \emph{maximal leakage}. 
\modified{\citet{IssaMaxL} define maximal leakage \modified{as the} \emph{average} information \modified{leaking to an adversary that aims to guess a randomized function of the secret.} Similarly, \citet{alvim2012MLviaGainFunctions} consider an adversary who aims to construct a guess of the secret that maximizes an arbitrary non-negative gain function. These two formulations can be shown to result in the same information leakage measure, that is, both operational definitions admit the same simplified quantity.}

While maximal leakage has \modified{these} strong operational \modified{foundations}, the fact that it is an on-average measure of information leakage may limit its applicability. Specifically, in \cite[Section 1]{pml}, the authors argue that averaging over all outcomes as done for maximal leakage \modified{may not provide sufficiently strong guarantees in privacy critical applications}. 
In addition, it was observed in \cite{IssaMaxL} that when the sensitive random variable $X$ takes values in an infinite set (e.g., the set of real numbers) maximal leakage can become infinite even in common scenarios such as adding Gaussian noise to Gaussian private data. To overcome these shortcomings, \citet{pml} propose \emph{pointwise maximal leakage} (PML), a generalization of maximal leakage. PML builds on similar threat models as maximal leakage via randomized functions and gain functions, but measures the information leaking about the private data $X$ at every realization of the public data $Y$ in isolation. As such, PML defines a random variable that describes the statistics of the information leaking about the private data \modified{and therefore} allows for highly flexible privacy guarantees: \modified{Various ways of assessing privacy can be expressed by considering different} statistics of \modified{the PML} random variable. 
Moreover, it is shown in \cite{saeidian2023pointwisegeneral} that PML can be used to make useful statements about the privacy of various systems in which maximal leakage becomes infinite, including the setup of adding Gaussian noise to a Gaussian random variable.  

Interestingly, \citet{saeidian2023inferential} also show that unlike differential privacy, PML provides clear guidelines for privacy parameter selection: For any given prior distribution, the chosen level of privacy determines the maximum amount of information (in terms of \emph{min-entropy}, that is, Rényi entropy of order infinity \cite{renyi1961entropy}) of any attribute of $X$ that can be disclosed by a privacy mechanism. This result constitutes a significant step towards the interpretability of privacy guarantees. Further, it offers a promising outlook on system design and policing, as privacy guarantees can be directly evaluated in terms of disclosure limits for each context. 

In order to design systems in accordance with PML, it is beneficial to provide optimal randomization strategies for achieving PML privacy, while keeping the privatized data as useful as possible for non-malicious inference. We will refer to this problem as the \emph{mechanism design problem}. The origin of privacy mechanism design can be attributed to the field of database privacy, in which (global) differential privacy \cite{dwork2014algorithmic} is by far the most prevalent measure used for trading off privacy against utility. In the global model, noise is added dynamically by query after the data is collected by a trusted curator, who has access to the complete private dataset. For this setup, perturbation mechanisms like the Laplacian and the Gaussian mechanisms have been shown to \modified{efficiently trade off privacy and utility} in various scenarios \cite{DPoriginalpaper,dwork2014algorithmic}. This paper deals with a \emph{local} model of privacy, that is, a model in which there is no trusted data curator, and privatization by randomization needs to be done locally (by each user) before releasing a data point. As the earliest example of a randomization strategy in such a local model, \citet{warnerRRoriginal} proposes a randomization strategy he refers to as the \emph{randomized response} technique. For simple binary cases, this technique has been shown to be optimal for local differential privacy and a broad class of convex utility functions \cite{extremalmechanismLong}. The discrete mechanism design problem has since been studied with many different privacy and utility measures, as well as more general source alphabets (see Section \ref{sec:intro:subsec:related}). 

In this paper, we explore the mechanism design problem for the local model with PML. We believe that the strong operational meaning of PML and its flexibility as well as its useful properties in terms of composition, pre-processing and post-processing~\cite{pml} make it a powerful framework for both analysis and design of private systems. Further, these properties give a promising perspective on privacy-by-design that is more easily aligned with \modified{formal} definitions as they are needed for effective legislature, as well as more holistic privacy in data-intensive applications.

\subsection{Overview and Contributions}
This paper presents various solutions to the mechanism design problem with PML, considering the sub-class of convex utility function previously presented in \cite{extremalmechanismLong}, which we call \emph{sub-convex} utility functions. 
Our proofs exploit general methods from convex analysis and majorization theory~\cite{majorization}.
\modified{
We briefly summarize our contributions as follows:
\begin{itemize}
\item \textbf{Cardinality bound}. We show that a mechanism maximizing sub-convex utility subject to a PML constraint does not need to increase the output alphabet size compared to the input.
\item \textbf{Characterization of the optimization problem}. We characterize the region of mechanisms satisfying $\epsilon$-PML for a fixed value of $\epsilon \geq 0$ as a convex polytope.
\item \textbf{Closed-form optimal mechanisms}. We present closed-form optimal mechanisms in the special cases of
\begin{enumerate}[label=(\roman*)]
    \item binary sensitive data,
    \item sensitive data with an arbitrary but finite alphabet in the \emph{high-privacy regime} (when the privacy parameter is sufficiently small), and
    \item uniformly distributed sensitive data. 
\end{enumerate}
\item \textbf{Optimal mechanisms via a linear program}. We present a linear program for computing optimal mechanisms in general scenarios. That is, the distribution of the sensitive data and the privacy level can both be picked arbitrarily.
\end{itemize}
}

\subsection{Other Related Works}
\label{sec:intro:subsec:related}
The privacy-utility tradeoff problem in the local setup has been studied in various works for different combinations of privacy and utility measures. To start with, mechanism design for the popular concept of LDP has been studied with utility measures such as Hamming distortion \cite{kalantari2018hamming,wang2016miprivacy}, minimax risk \cite{duchi2013LDPminmaxDEF} and the previously mentioned sub-convex functions \cite{extremalmechanismLong}, which include, e.g., mutual information. While LDP is not context-aware, a context-aware framework for mechanism design with LDP has been proposed in~\cite{acharya2020context}. The privacy-utility tradeoff has also been extensively studied using information theoretic measures. For example, \citet{8849440} present what they call a \emph{watchdog mechanism} that leverages \modified{LIP} to evaluate the risk of a privacy breach any data sample presents and adapts the privatizing randomization strategy accordingly. In \cite{9965910}, these watchdog mechanisms are adapted to satisfy an extension of \modified{LIP} to further enhance utility. Other \modified{LIP} mechaniss are designed to minimize expected distortion in \cite{jiang2020LIP,jiang2021LIPcontextaware} \modified{and linear distance measures in \cite{9517826}}. A linear programming approach for designing optimal LIP mechanisms is presented in \cite{lopuhaa2020data}. Further, mechanisms for maximal leakage are designed with utility measures like Hamming distortion~\cite{saeidian2021hamming}, upper triangular cost matrices~\cite{9162276}, and the Type-II error exponent in a hypothesis testing framework~\cite{8006634}. 

\modified{
\subsection{Outline  of the Paper}
The rest of the paper is organized as follows: In Section \ref{sec:prelim} we will review the definition of PML, as well as some basic definitions in majorization theory. In Section \ref{sec:tradeoffproblem} we will present and analyze the general optimization problem considered in this paper. Section \ref{sec:optsol} presents the results on optimal mechanism design. Section \ref{sec:conclusions} concludes the paper. 
}

%% file: sections/preliminaries.tex
\label{sec:prelim}
\subsection{Notation}
\label{sec:prelim:notation}
Generally, lowercase boldface letters denote vectors, while uppercase boldface letters denote matrices, e.g., $\boldsymbol{p} \in \mathbb{R}^N$, $\boldsymbol{Q} \in \mathbb{R}^{N\times N}$. We denote the $j^{\text{th}}$ column of a matrix $\bm P$ as $\bm P_j$. Single elements of a matrix $\bm P$ are denoted by the corresponding lower-case indexed letter $p_{ij}$. \modified{$\boldsymbol{I}_N$ denotes the identity matrix of size $N$.} This paper focuses on finite random variables, and as a result, all sets are assumed to be finite. Random variables are represented using uppercase letters, such as $X$, while uppercase calligraphic letters represent sets, such as the alphabet of $X$, which is denoted by $\mathcal{X}$. Given random variables $X$ and $Y$, $P_{XY}$ is used to indicate their joint probability distribution, while $P_X$ and $P_Y$ denote the marginal distributions of $X$ and $Y$, respectively. The conditional probability kernel $P_{Y|X}$ is referred to as the \emph{privacy mechanism}. We assume that $|\mathcal{X}| = N$ and $|\mathcal{Y}|=M$. 
We use $\supp(P_X) = \{x \in \mathcal{X}: P_X(x) > 0\}$ to denote the support set of the distribution $P_X$. Unless stated otherwise, we assume a random variable to have full support on its alphabet, that is, $\supp(P_X) = \mathcal{X}$. For notational convenience, we assume that the set $\mathcal{X} = \{x_1,\dots,x_N\}$ is ordered in non-increasing probability, that is, $P_X(x_1) \geq \dots \geq P_X(x_N)$. \modified{Finally,} we use $[N]$ to denote the set of positive integers up to $N$, that is, $[N] \coloneqq \{1,\dots,N\}$.

\subsection{Pointwise Maximal Leakage}
\label{sec:prelim:pml}

We consider the random variable $X$ to be the private data. A mechanism $P_{Y|X}$ then privatizes (that is, randomizes) this sensitive data and outputs a sanitized view of $X$, denoted by $Y$. \modified{We measure the amount of information each outcome $Y=y$ leaks about the private data $X$} using the pointwise maximal leakage (PML) measure proposed by~\citet{pml}. 

Although it has two equivalent operational definitions via randomized functions and generalized gain functions, PML admits a simple formulation. We start by introducing the operational formulations.
\begin{definition}[Pointwise maximal leakage (PML), {{\cite{pml}}}]
    Let $P_{XY}$ be the joint distribution of two random variables defined on the finite set $\mathcal{X} \times \mathcal{Y}$. Suppose the Markov chain $U - X - Y - \hat{U}$ holds. Then the \emph{pointwise maximal leakage from $X$ to an outcome $y \in \mathcal{Y}$} is defined as
    \begin{equation}
    \label{eq:PMLdefRandFunc}
            \ell(X \to y) \coloneqq 
    \log \sup_{P_{U \mid X}}\frac{\sup\limits_{P_{\hat U \mid Y=y}} \mathbb{P} \left[U=\hat U \mid Y=y \right]}{\max\limits_{u\in \mathcal{U}} P_U(u)}.
    \end{equation}
\end{definition}
In this definition, information leakage is measured by the relative increase in the probability of correctly guessing an attribute $U$ of the private data $X$ when observing $Y=y$, compared to a ``blind" guess made without observing $Y$. As shown in \cite{pml}, this formulation is equivalent to another operational formulation: Assume an adversary picks her guess of $X$ from a non-empty set $\hat{\mathcal{X}}$. Assume further that she measures the gain she gets from the guess via a function $g: \mathcal{X} \times \hat{\mathcal{X}} \to \mathbb{R}_+$. Then the randomized function view of PML in \eqref{eq:PMLdefRandFunc} can be shown to be equivalent to the worst-case increase in expected gain the adversary gets from observing $Y$, that is,
\begin{equation}
    \ell(X \rightarrow y) = \log \sup_{g}
    \frac{\sup\limits_{P_{\hat{X} \mid Y=y}} \mathbb{E} \left[g(X,\hat{X}) \mid Y=y \right]}{\sup\limits_{\hat x \in \hat{\mathcal X}} \mathbb E\left[g(X, \hat{x})\right]}.
\end{equation}
In \cite[Theorem 1]{pml}, it is shown that in the case of finite alphabets, these equivalent definitions are given by the maximum information density $i_{P_{XY}}(x;y) \coloneqq \log \frac{P_{XY}(x,y)}{P_X(x) \, P_Y(y)}$ of the joint distribution $P_{XY}$ of $X$ and $Y$ considering all outcomes of $X$:
\begin{align}
\label{eq:simplyfiedPML}
    \ell(X\rightarrow y) &= \log \max_{x \in \supp(P_X)} \frac{P_{Y|X=x}(y)}{P_Y(y)}\\ &= \max_{x \in \supp(P_X)}i_{P_{XY}}(x;y).
\end{align}
\modified{We always have $\ell(X \to y) \geq 0$. In the finite alphabet case, assuming a fixed $P_X$}, the PML is also upper bounded by $\ell(X\rightarrow y) \leq -\log(\min_{x \in \supp(P_X)}P_X(x))$, implying that \modified{it} remains finite. We use $\epsilon_{\max} \coloneqq  -\log(\min_{x \in \supp(P_X)}P_X(x))$ to denote this upper bound. 

Since PML is defined separately for each outcome $y$, the leakage $\ell(X \to Y)$ becomes a random variable when considering $Y \modified{\sim P_Y}$. In order to provide a strict privacy guarantee, we consider the \emph{almost-sure guarantee} \cite[Definition 4]{pml}: This definition bounds the leakage of all outcomes of $Y$ as $\mathbb{P}_{Y \sim P_Y}[\ell(X \rightarrow Y) \leq \epsilon] = 1$. Any mechanism satisfying this property is said to satisfy \emph{$\epsilon$-PML}. From a design perspective, this is equivalent to restricting the leakage of each outcome of $Y$ separately to be smaller than the required privacy level $\epsilon$. Obviously, all mechanisms satisfy $\epsilon_{\mathrm{max}}$-PML.


\subsection{Majorization Theory}
\label{sec:prelim:majorization}
In this section, we will restate a few key definitions of majorization theory. Majorization theory provides a partial order on sets of elements with equal cardinality and equal sum, and can therefore be seen as a way of measuring the ``uniformity" of a pmf. In the context of this paper, we will leverage majorization theory to analyze the behavior of privacy guarantees concerning the data's prior distribution. For a detailed discussion on majorization theory, we refer to \cite{majorization}. All statements listed below can be found there.
\begin{definition}[Majorization]
Given a tuple $x \in \mathbb{R}^N$ and $i \in \{1, \ldots, N\}$, denote by $x_{(i)}$ the $i^\text{th}$ largest element of $x$. Consider two tuples $p, q \in \mathbb{R}^N$. We say that $p$ \emph{majorizes} $q$, written as $p \succ q$, if $\sum_{i=1}^N p_i = \sum_{i=1}^N q_i$, and for all $k = 1,\dots,N-1$: $\sum_{i=1}^k p_{(i)} \geq \sum_{i=1}^k q_{(i)}$.
\end{definition}
As an example, if $p=(\nicefrac{1}{3},\nicefrac{1}{3},\nicefrac{1}{3})$ and $q=(\nicefrac{2}{3},\nicefrac{1}{3},0)$, then $q \succ p$.

\begin{definition}[Schur-convex / Schur-concave function]
A function $\phi: \mathbb{R}^N \rightarrow \mathbb{R}$ is said to be \emph{Schur-convex}, if $p \succ q  \, \Rightarrow \phi(p) \geq \phi(q)$. Further, $\phi(p)$ is said to be \emph{Schur-concave} if and only if $-\phi(p)$ is Schur-convex.
\end{definition}
For example, $\max (\cdot)$ is a Schur-convex function while $\min (\cdot)$ and the Rényi entropy~\cite{renyi1961entropy} are Schur-concave.

%% file: sections/tradeoff_problem.tex
\label{sec:tradeoffproblem}
The aim of this section is twofold: \modified{Firstly, we present important results needed for finding optimal mechanisms in the PML framework. 
We prove a cardinality bound on the output alphabet of the optimal mechanisms and show that we can without loss of generality assume an optimal mechanism to have at most full output support, i.e., $\supp(P_Y) \leq \supp(P_X)$. Using this fact, we fully characterize the optimization region as a polytope, one of whose vertices constitutes an optimal solution to the privacy-utility tradeoff problem. In the second part of this section, we utilize the PML framework to analyze the prior-distribution dependence of the privacy guarantees provided by the randomized response mechanism optimized for LDP. 
This will enable a more realistic comparison of mechanism performances, as the results allow us to pick the parameter of the randomized response mechanism to exactly achieve a privacy guarantee specified in terms of $\epsilon$-PML.}

\subsection{The Mechanism Design Problem}
\label{sec:tradeoffproblem:mechdesign}
We consider a general discrete privacy mechanism $P_{Y \mid X}$ mapping $N$ input symbols to $M$ sanitized output symbols. For simplicity, we use a matrix
\begin{equation}
\label{eq:generalmechansim}
    \bm P = \begin{bmatrix}
        p_{11} & \dots & p_{1M} \\
        \vdots & \ddots & \vdots \\
        p_{N1} & \dots & p_{NM}
    \end{bmatrix} \in \mathbb [0,1]^{N \times M},
\end{equation}
to represent the privacy mechanism, where $p_{ij} \coloneqq P_{Y|X=x_i}(y_j)$. Evidently, in order to \modified{form} a valid transitioning kernel, this matrix needs to be row-stochastic, and its elements need to be bounded by $0 \leq p_{ij} \leq 1$. In Section~\ref{sec:tradeoff:subsec:optregion}, we will derive more detailed constraints on $\bm P$ including the ones imposed by PML. We also use $\bm P_j$ with $j \in [M]$ to denote the $j^{\mathrm{th}}$ column of $\bm P$. 
 
\subsubsection{Sub-convex utility functions}
We measure the utility of the privatized data using a \modified{rich} sub-class of convex functions previously studied in \cite{extremalmechanismLong}, which we will refer to as \emph{sub-convex} utility functions.

\begin{definition}[Sub-convex function, \cite{extremalmechanismLong}]
\label{def:row-sum-sub-linear}
A function $U: \mathbb R_+^{N \times M} \to \mathbb R_+$ is said to be \emph{sub-convex} if it has the form 
\begin{equation}
U(\boldsymbol{P}) = \sum_{j =1}^M\mu(\bm P_j),
\end{equation}
where $\mu : \mathbb R_+^N \to \mathbb R_+$ is a sub-linear function.\footnote{A function $\mu: \mathbb R_+^N \to \mathbb  R_+$ is said to be sub-linear if $\forall \, \lambda \in \mathbb R_+$ and $\forall \bm x, \bm y \in \mathbb R^N_+$ we have $\mu(\lambda \bm x) = \lambda \mu(\bm x)$ and $\mu(\bm x + \bm  y) \leq \mu(\bm x) + \mu(\bm y)$. These properties together also imply convexity.}
\end{definition}

\modified{It is shown in \cite{extremalmechanismLong} that sub-convex functions according to Definition \ref{def:row-sum-sub-linear} satisfy a data processing inequality. This property will be needed for proving optimality in Theorems \ref{thrm:highprivacymechanism} and \ref{thrm:allregionseqprioropt}.}
The class of sub-convex functions includes, \modified{e.g., any $f$-divergence between marginal distributions induced by two candidate hypotheses, as well as any $f$-information between private and public data. In what follows,} the main instance of \modified{sub-convex} functions that we will use for illustrations is the mutual information $I(X;Y)$ between the private and the released data, defined as the Kullback-Leibler divergence between the joint distribution and the marginals of the two random variables, that is, $I(X;Y) = D_{\text{KL}}(P_{XY}\|P_X P_Y)$. 

\begin{remark}
\label{remark:equivalenceclass}
It is clear by Definition~\ref{def:row-sum-sub-linear} that column permutations of a mechanism $\bm P$ do not affect its utility. Formally, given row-stochastic matrices $\bm P, \bm Q \in \mathbb [0,1]^{N \times M}$, we may define an equivalence relation $\bm P \sim \bm Q$, where $\bm Q$ is obtained by permuting the columns of $\bm P$. Then, all mechanisms in the equivalence class $[\bm P]$ achieve the same utility. 
\end{remark}

\subsubsection{Optimization problem formulation}
We now present the \modified{fundamental} optimization problem considered in this paper. For a fixed $P_X$, we define 
\begin{equation}
    \epsilon_m (\bm P) \coloneqq \inf \Big\{\epsilon \geq 0 \colon \mathbb P\left[\ell(X \to Y) \leq \epsilon \right]=1 \Big\},
\end{equation}
to be the smallest value of $\epsilon \geq 0$ at which the mechanism $\bm P$ satisfies $\epsilon$-PML. Let $\mathcal S_{N,M} \subset \mathbb [0,1]^{N \times M}$ denote the set of all $N \times M$ row-stochastic matrices. Given $\epsilon \geq 0$, we then define  
\begin{equation}
    \mathcal M(\epsilon) \coloneqq \left\{\bm P \in  \bigcup_{M=1}^\infty \mathcal S_{N,M} \colon \epsilon_m(\bm P) \leq \epsilon \right\},
\end{equation}
to be the set of all privacy mechanisms with $N$ input symbols (i.e., rows) that satisfy $\epsilon$-PML. Then, our privacy-utility tradeoff problem can be expressed as 
\begin{equation}
\label{eq:generalOPTproblem}
    U^*(\epsilon) \coloneqq \sup_{\bm P\in \mathcal M(\epsilon)} U(\bm P), 
\end{equation}
where our goal is to find the largest utility $U^*(\epsilon)$ for a fixed privacy parameter $\epsilon$. We use both $P^*_{Y \mid X}$ and $\bm P^*$ to denote the optimal mechanism in the above problem.

\subsection{Cardinality Bound}
\label{sec:optsol:subsec:cardinality}
\modified{Problem~\eqref{eq:generalOPTproblem} does not make any assumptions on the size of $\mathcal Y$. In principle, we yet have no reason to assume that utility could not increase with increasing $M=|\mathcal Y|$.} To solve problem~\eqref{eq:generalOPTproblem}, we therefore first show that the search for an optimal mechanism can be restricted to mechanisms that do not increase the output alphabet size compared to the input.
\begin{theorem}[Cardinality bound]
\label{thrm:cardinality}
    To solve the optimization problem in \eqref{eq:generalOPTproblem}, it suffices to consider mechanisms $P_{Y|X}$ such that $|\mathcal{Y}| \leq |\mathcal{X}|$ holds.
\end{theorem}
The proof of this theorem, which is based on the \emph{perturbation method} \modified{\cite{gohari2012evaluation,el2011networkinformationtheory}}, is deferred to Appendix~\ref{app:proofOfCardinality}. Equipped with this theorem, we can now assume that $N \geq M$. Since all such $N \times M$ mechanisms can be written as an $N \times N$ mechanism that contains $N-M$ all-zero columns, without loss of generality in the rest of the paper we restrict our attention to $N \times N$ row-stochastic matrices. More formally, we define 
\begin{equation}
    \mathcal M_N(\epsilon) \coloneqq \left\{\bm P \in  \mathcal S_{N,N} \colon \epsilon_m(\bm P) \leq \epsilon \right\},
\end{equation}
to be the set of all $N \times N$ privacy mechanisms that satisfy $\epsilon$-PML. Then, by Theorem~\ref{thrm:cardinality}, for all $\epsilon \geq 0$ we have 
\begin{equation}
\label{eq:optim_prob_cardinality}
    U^*(\epsilon) = \max_{\bm P \in \mathcal M_N(\epsilon)} U(\bm P). 
\end{equation}

\subsection{Characterization of the Optimization Region}
\label{sec:tradeoff:subsec:optregion}

In order to obtain methods for efficiently computing optimal mechanisms, it is useful to express the privacy constraint as a collection of linear inequalities. 
\begin{lemma}
\label{lemma:polytope}
Given any privacy level $\epsilon \geq 0$ and a prior distribution $P_X$, the set $\mathcal{M}_N(\epsilon)$ is a closed and bounded polytope in $[0,1]^{N \times N}$, described by the linear constrains 
\begin{subequations}
\label{eq:constraints}
\begin{align}
        p_{ij} - \biggl(\sum_{i=1}^N P_X(x_i) p_{ij}\biggr)\e^{\epsilon} \leq 0, \quad  \forall i,j \in [N],
        \label{subeq:pmlconstraints} \\
        \sum_{j=1}^N p_{ij} = 1, \quad \quad \forall i \in [N],
        \label{subeq:sumconstraint} \\
        p_{ij} \geq 0, \quad  \forall i,j \in [N].
        \label{subeq:nonnegativityconstraint}
\end{align}
\end{subequations}    
\end{lemma}
The proof of this Lemma is provided in Appendix \ref{app:polyproof}. In the following discussions, we will refer to the constraints $0 \leq p_{ij} \leq 1$ for all $i,j \in [N]$ implied by \eqref{subeq:sumconstraint} and \eqref{subeq:nonnegativityconstraint} as the \emph{box constraints}, while \eqref{subeq:pmlconstraints} will be referred to as the \emph{PML constraints}. 
Since our utility functions are convex, they will be maximized at an extreme point of $\mathcal{M}_N(\epsilon)$. Inspired by \cite{extremalmechanismLong}, we will refer to the vertices of $\mathcal{M}_N(\epsilon)$ as \emph{extremal mechanisms}. Note that we may use standard methods from linear programming to enumerate all extreme points of $\mathcal{M}_N(\epsilon)$ by finding the basic feasible solutions of the given linear constraints (see \cite{Griva2009}). \modified{In methodology, this approach also shares similarity with the vertex enumeration approach presented in \cite{lopuhaa2020data} for LIP}.\footnote{\modified{We remark that LIP imposes both lower and upper bound on the information density values, while \eqref{eq:simplyfiedPML} only involves an upper bound. This additional lower bound in LIP results in a significantly different behavior with respect to mechanism design. See \cite{grosse2024quantifying} for a detailed discussion.}} While \modified{these vertex enumerations} can be directly implemented, \modified{their} computational complexity grows significantly for larger alphabet sizes. \modified{This issue is addressed in \cite{lopuhaa2020data} by presenting algorithmic data release protocols. In contrast, in this work} we set out to find closed-form optimal solutions under various assumptions \modified{in Section \ref{sec:optsol}}. 



\subsection{\modified{Relation to Randomized Response for LDP}}
 \input{sections/rr_analysis}

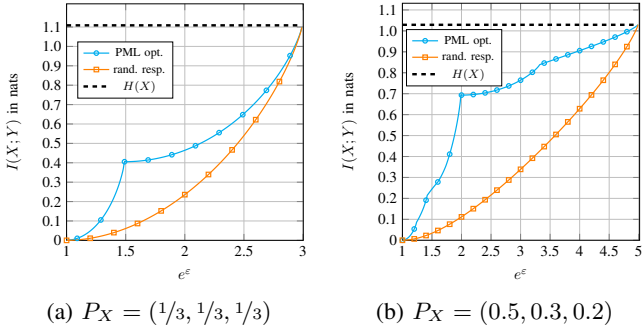
\begin{figure}[!t]
\centering
\begin{subfigure}[b]{.49\columnwidth}
    \centering
    \resizebox{\textwidth}{!}{
    \input{fig/utilcomparison_unif}
    }
    \caption{$P_X = (\nicefrac{1}{3},\nicefrac{1}{3},\nicefrac{1}{3})$}
    \label{fig:utilcompN3:subfig:unif}
\end{subfigure}
\begin{subfigure}[b]{.49\columnwidth}
    \centering
    \resizebox{\textwidth}{!}{
    \input{fig/utilcomparison_nonunif}
    }
    \caption{$P_X= (0.5,0.3,0.2)$}
    \label{fig:utilcompN3:subfig:nonunif}
\end{subfigure}
\caption{Comparison of mutual information utility between the PML-optimal mechanisms \modified{and the randomized response mechanisms in \eqref{eq:LDP-RR} for $N=3$}. The dashed line marks the maximum utility, i.e., the Shannon entropy of $X$.}
\label{fig:utilcompN3}
\end{figure}

In Figure \ref{fig:utilcompN3}, we \modified{plot} the \modified{mutual information achieved by} the PML optimized mechanisms \modified{compared to the randomized response mechanism. The optimal PML mechanism is computed via a straightforward vertex enumeration of $\mathcal M_N(\epsilon)$ and a subsequent exhaustive search for the vertex maximizing utility.} Given \modified{a value of $\epsilon < \epsilon_{\max}$}, \modified{in order to find the randomized response mechanism exactly achieving $\epsilon$-PML}, the value of $\epsilon_r$ is obtained \modified{from} \eqref{eq:epspml} as
 \begin{equation}
 \label{eq:epsldpofepspml}
     \epsilon_r(\epsilon) = \epsilon + \log \frac{1-p_{\min}}{1- p_{\min}\e^{\epsilon}}.
 \end{equation}
As expected, the plots show that the PML-optimized mechanisms are able to more efficiently exploit the privacy budget to achieve higher utility.
Figure \ref{fig:utilcompN3:subfig:unif} depicts the case of a uniform prior on $X$. Note that the characteristic point at $\epsilon = \log 1.5 = \log \frac{1}{P_X(x_1) + P_X(x_2)}$ marks the transition between the two \emph{privacy regions}, as described in Section~\ref{sec:prelim:subsec:regions}. Figure \ref{fig:utilcompN3:subfig:nonunif} considers \modified{a non-uniform prior} case \modified{with} $P_X = (0.5, 0.3, 0.2)$.

%% file: sections/rr_analysis.tex
\label{sec:tradeoff:subsec:rr}
\modified{To contrast the extremal PML mechanisms against existing solutions for LDP, we present} an analysis of the popular \emph{randomized response mechanisms} in the PML framework. These mechanisms are in many scenarios optimal for the class of sub-convex utility functions \modified{subject to LDP constraints} \cite{extremalmechanismLong}. The results show that in the PML framework, mechanisms can exploit knowledge about the data's prior distribution in order to increase utility compared to the randomized response mechanism, even under \modified{similar} leakage requirements. 

\begin{definition}[\modified{Randomized response} {{\cite{extremalmechanismLong}}}]
\label{def:LDP-RR}
    Given a source alphabet $\mathcal{X}$ of size $|\mathcal{X}| = N$, the \emph{randomized response mechanism} \modified{with parameter $\epsilon_r \geq 0$} is given by
    \begin{equation}
    \label{eq:LDP-RR}
        P_{Y|X=x_i}(y_j) = \begin{cases}\frac{e^{\epsilon_r}}{(N-1)+e^{\epsilon_r}},& j=i \\
                              \frac{1}{(N-1)+e^{\epsilon_r}}, & j\neq i
                              \end{cases} \quad \forall i,j \in [N].
    \end{equation}
\end{definition}

 The \modified{PML} of a randomized response mechanism to $y_j \in \mathcal{Y}$ \modified{is calculated by substituting the maximum value of the conditional distribution \eqref{eq:LDP-RR} at $i=j$ into \eqref{eq:simplyfiedPML}.\footnote{\modified{The case $j=i$ will always maximize $P_{Y|X=x_i(y_j)}$, as $\epsilon_r \geq 0$.}} Note that, since the conditional distribution is given by a doubly-stochastic matrix, we also have $P_Y = P_X$. This yields} 
\begin{equation}
\label{eq:RRpml}
    \ell(X\to y_j) = \epsilon_r - \log \biggl( (\e^{\epsilon_r} - 1)P_X(x_j) + 1 \biggr), \quad \forall j \in [N].
\end{equation}

From \eqref{eq:RRpml}, it becomes clear that the \modified{PML of the mechanisms in~\eqref{eq:LDP-RR}} will be different from the privacy parameter \modified{$\epsilon_r$} corresponding to the mechanism's LDP guarantee. More specifically, \modified{analyzing the} randomized response mechanism \modified{in the PML framework} shows that\modified{, in this context, the \modified{largest} amount of information leaked to an adversary depends} on the \modified{minimum probability of the symbols in $X$. By maximizing \eqref{eq:RRpml} over all $y_j$, we obtain}
\begin{align}
\label{eq:epspml}
    \epsilon(P_X) &:= \underset{i\in [N]}{\max}\,\log\, \biggl(\frac{\e^{\epsilon_r}}{P_X(x_i)(\e^{\epsilon_r}-1) + 1}\biggr) \\ &= \log \, \biggl(\frac{\e^{\epsilon_r}}{P_X(x_N)(\e^{\epsilon_r}-1) + 1}\biggr).
\end{align}
The following proposition characterizes this leakage as a function of the prior distribution of $X$.
\begin{proposition}
\label{prop:epspmlschurconvex}
The mapping $P_X \mapsto \epsilon(P_X)$ is Schur-convex.
\end{proposition}
\begin{proof}
It is enough to note that \eqref{eq:epspml} can be seen as the composition $\epsilon = h(\phi(P_X(x)))$ by setting $h(t) := \log(\frac{\e^{\epsilon_r}}{t(\e^{\epsilon_r}-1)+1})$ and $\phi(P_X(x_1),\dots,P_X(x_N)) := \underset{i}{\min}\,P_X(x_i)$. Clearly, $h(t)$ is decreasing in $t \in \mathbb{R}$. Further, as shown in \cite{majorization}, $p \mapsto \min(p)$ is Schur-concave for $p \in \mathbb{R}^N$. By \cite[Table 1]{majorization}, compositions of this form are Schur-convex.
\end{proof}

This result provides insights into the behavior of \modified{the} randomized response mechanisms \modified{w.r.t.} \modified{PML}. Note that \modified{PML} is upper bounded by \modified{the LDP privacy parameter} \cite[Proposition 6]{pml}. When $P_X(x_j)$ tends to $1$ for some $j$, we get $\ell(X \rightarrow y_j) \rightarrow 0$, and $\ell(X \rightarrow y_i) \rightarrow \epsilon_r$ for all $i \neq j$. Due to \modified{Proposition \ref{prop:epspmlschurconvex}}, this can be seen as the worst-case prior distribution.

%% file: fig/utilcomparison_unif.tex
\pgfplotstableread[col sep=comma,]{fig/utilcomparisonRR-PML_N-3_unif.csv}\datatable
\begin{tikzpicture}
\begin{axis}[scale=2,ymin=0.0,xmax=182, xlabel=$\e^{\epspml}$,ylabel= $I(X;Y)$ in nats,every axis x label/.style={at={(current axis.right of origin)},anchor=north},
               width=1*\textwidth, height=1*\textwidth, every axis plot post/.append style={domain=0:5,samples=50,smooth},
  enlargelimits=upper,
  xtick={0,50,100,150,200},
  xticklabels={1,1.5,2,2.5,3},
  xticklabel style={},
  ytick={0,0.1,0.2,0.3,0.4,0.5,0.6,0.7,0.8,0.9,1.0,1.1},
  yticklabels={0,0.1,0.2,0.3,0.4,0.5,0.6,0.7,0.8,0.9,1.0,1.1},
  x label style={at={(axis description cs:0.5,-0.1)},anchor=north},
  y label style={at={(axis description cs:0.15,.5)},rotate=90,anchor=south},
  ylabel style={rotate=-90,anchor=south,yshift=1.5em},
  legend style={at={(0.45,0.6)},anchor=south east,font=\footnotesize},
  mark size = 2.0pt, xmajorgrids,ymajorgrids]    
    \addplot [cyan,thick,mark=o,mark size=1.5pt, mark repeat=20, mark phase=10] table [x expr=\coordindex, y={PML}]{\datatable};
    \addlegendentry{PML opt.}
    
    \addplot[orange,thick,mark=square,mark size=1.5pt,mark repeat=20] table [x expr=\coordindex, y={RR}]{\datatable};
    \addlegendentry{rand. resp.}
    
    \addplot[black, dashed, ultra thick] table [x expr=\coordindex, y={H}]{\datatable};
    \addlegendentry{$H(X)$};

\end{axis}
\end{tikzpicture}

%% file: fig/utilcomparison_nonunif.tex
\pgfplotstableread[col sep=comma,]{fig/utilcomparisonRR-PML_N-3_0p5_0p3_0p2.csv}\datatable
\begin{tikzpicture}
\begin{axis}[scale=2,ymin=0.0,xmax=365, xlabel=$\e^{\epspml}$,ylabel= $I(X;Y)$ in nats,every axis x label/.style={at={(current axis.right of origin)},anchor=north},
               width=1*\textwidth, height=1*\textwidth, every axis plot post/.append style={domain=0:5,samples=50,smooth},
  enlargelimits=upper,
  xtick={0,50,100,150,200,250,300,350,400},
  xticklabels={1,1.5,2,2.5,3,3.5,4,4.5,5},
  xticklabel style={},
  ytick={0,0.1,0.2,0.3,0.4,0.5,0.6,0.7,0.8,0.9,1.0,1.1},
  yticklabels={0,0.1,0.2,0.3,0.4,0.5,0.6,0.7,0.8,0.9,1.0,1.1},
  x label style={at={(axis description cs:0.5,-0.1)},anchor=north},
  y label style={at={(axis description cs:0.15,.5)},rotate=90,anchor=south},
  ylabel style={rotate=-90,anchor=south,yshift=1.5em},
  legend style={at={(0.42,0.65)},anchor=south east,font=\footnotesize},
  mark size = 2.0pt, xmajorgrids,ymajorgrids]    
    \addplot [cyan,thick,mark=o,mark size=1.5pt,mark repeat=20] table [x expr=\coordindex, y={PML}]{\datatable};
    \addlegendentry{PML opt.}
    
    \addplot[orange,thick,mark=square,mark size=1.5pt,mark repeat=20] table [x expr=\coordindex, y={RR}]{\datatable};
    \addlegendentry{rand. resp.};
    
    \addplot[black, dashed, ultra thick] table [x expr=\coordindex, y={H}]{\datatable};
    \addlegendentry{$H(X)$};

\end{axis}
\end{tikzpicture}

%% file: sections/opt_sol.tex
\label{sec:optsol}
In this section, we derive optimal mechanisms for the privacy-utility tradeoff problem \eqref{eq:generalOPTproblem} with PML and sub-convex utility functions. We present various closed-form solutions \modified{under different assumptions on $\epsilon$ and $P_X$}. \modified{To begin with, we present the optimal binary randomization strategy for all $\epsilon$ and any arbitrary but fixed prior distribution of binary private data $X$. Then, we introduce the notion of privacy regions, defined based on disclosure limits of PML guarantees. We proceed to present an optimal mechanism for general alphabet sizes in the \emph{high-privacy} regime, the privacy region with the strictest disclosure prevention guarantee (i.e., when $\epsilon$ is sufficiently small). Then, we present optimal mechanisms for all privacy regions and a uniform prior distribution on $X$. We also present a linear program that is able to efficiently compute optimal mechanisms for all privacy levels and any arbitrary prior distributions on $X$. Finally, we provide illustrative numerical simulations demonstrating the presented mechanisms.}

%% file: sections/binary.tex
\label{sec:optsol:subsec:binary}
\modified{First, we focus on mechanisms for binary input alphabets.} Although this binary setup is, in terms of cardinality, a minimal example, it is particularly useful for two reasons. First, in many privatization scenarios, the sensitive features are indeed binary. Consider any survey asking participants about personal details with yes/no answers, e.g., sex (in the biological sense), the presence (or absence) of a specific disease, and so on. Second, the binary case admits a closed-form solution that is analytically tractable. Since the presented solution covers all combinations of privacy levels and prior distributions, it makes the behavior of PML-optimal mechanisms explainable and thereby provides insights about the privacy measure.

\begin{theorem}
\label{thrm:binarymechanism}
Suppose $X$ is distributed according to $P_X$ and let $\mathcal{X} \coloneqq \{x_1,x_2\}$. Given $0 \leq \epsilon \leq \epsilon_{\max}$, \modified{an instance of} the optimal privacy mechanisms \modified{$[\bm P^*]$} in problem \eqref{eq:optim_prob_cardinality} is
\begin{align}
\label{eq:binarymech}
    &P^*_{Y|X} = \\
    &\begin{dcases}\modified{\begin{bmatrix}
        e^{\epsilon}P_X(x_2) & 1-e^{\epsilon}P_X(x_2) \\
        1-e^{\epsilon}P_X(x_1) & e^{\epsilon}P_X(x_1)
    \end{bmatrix}} & \textnormal{if} \; P_X(x_1) < e^{-\epsilon}, \\
    \begin{bmatrix}
        \frac{e^{\epsilon}-1}{e^{\epsilon}P_X(x_1)} & \frac{1-e^{\epsilon}P_X(x_2)}{e^{\epsilon}P_X(x_1)} \\
        0 & 1
    \end{bmatrix} & \textnormal{if} \; P_X(x_1) \geq e^{-\epsilon}. \nonumber
    \end{dcases}
\end{align}
\end{theorem}
The proof of this theorem is provided in Appendix \ref{app:binaryproof}.

\subsubsection{Interpretation and Intuitive Insights}

The condition \modified{$P_X(x_1) \geq e^{-\epsilon}$} has a clear interpretation: It \modified{describes a scenario} in which \modified{the symbol $x_1$} has a relatively large probability and, as a result, the outcome \modified{$X = x_1$ can be deterministically disclosed}. Since in the PML framework, we assume that the adversary knows $P_X$, the probability of correctly guessing this symbol is already high \emph{a priori}\modified{, and hence the prior-to-posterior ratio in the PML formulation is small, even for a deterministic release of this symbol.} However, for the less probable outcome, a correct guess after observing $Y$ constitutes a large increase in leakage \modified{because} without observing $Y$, the adversary is unlikely to guess this outcome correctly. The corresponding channel therefore only \modified{masks} the less probable outcome of $Y$ in order to reach the desired privacy level. \modified{Any adversary will then be able to deterministically infer the realization of $X$ if she observes $Y=y_1$. On the other hand, the randomization in the channel permits such a deterministic inference given that $Y=y_2$.}

%% file: sections/regions.tex
\label{sec:prelim:subsec:regions}
As discussed in Section~\ref{sec:prelim:pml}, the parameter $\epsilon$ in an $\epsilon$-PML guarantee is always finite and bounded by $\epsilon_{\mathrm{max}}$ and it is easy to see that the mechanism $\bm P = \bm I_N$ satisfies $\epsilon_{\mathrm{max}}$-PML. Hence, unlike \modified{, e.g.,} LDP, PML allows zero probability assignments in the mechanism \modified{matrix} $\bm P$. That is, depending on $\epsilon$, there may be an outcome $y$ with $P_Y(y) > 0$ such that $P_{Y \mid X=x}(y) =0$ for some $x \in \mathcal X$. Thus, in order to find closed forms for optimal mechanisms, first we need to understand how the value of $\epsilon$ determines the number of zeros in each column of matrix $\bm P$. 

\modified{Recall that we assume the symbols in $\mathcal X$ to be ordered by non-increasing probability.} Let $\epsilon_k(P_X) \coloneqq -\log \sum_{i=1}^{N-k}P_X(x_i)$ for $k \in \{0, \ldots, N-1\}$. We say that $\epsilon$ is in the \emph{$k^{\text{th}}$ privacy region} if $\epsilon \in [\epsilon_{k-1}(P_X), \epsilon_k(P_X))$, where $k \in [N-1]$. \modified{We have} the following lemma, which is a generalization of \cite[Proposition 4]{saeidian2023inferential}.

\begin{lemma}
\label{lem:privacyregiondisclosure}
Suppose $X$ is distributed according to $P_X$. 
If $\epsilon$ is in privacy region $k \in [N-1]$, then each column of a mechanism $\bm P \in \mathcal M_N(\epsilon)$ can contain at most $k-1$ zero entries. 
\end{lemma}
The proof of this lemma is provided in Appendix \ref{app:proofLemLowPrivacy}. We emphasize that Lemma~\ref{lem:privacyregiondisclosure} does \emph{not} mean that in privacy region $k$ any arbitrary collection of $k-1$ elements in a column can be set to zero. Instead, Lemma~\ref{lem:privacyregiondisclosure} states that if a privacy mechanism includes a column with $k -1$ zero entries, then $\epsilon \geq \epsilon_{k-1}(P_X)$. \modified{The statement in Lemma \ref{lem:privacyregiondisclosure} therefore implies a definition of privacy regions that are ordered from strictest ($\epsilon \in [0,\epsilon_1)$) to least strict ($\epsilon \in [\epsilon_{N-1},\epsilon_{\max})$), where, conceptually, strictness is defined by the number of possible zero-assignments in the mechanism matrix.}


%% file: sections/high_privacy.tex
\label{sec:optsol:subsec:highprivacy}

We will now consider a specific privacy region in greater detail, which we will call the \emph{high-privacy regime}.
The high-privacy regime for a given prior distribution $P_X$ is defined as the values of the privacy parameter $\epsilon$ that fall into the \emph{first} privacy region. By Lemma~\ref{lem:privacyregiondisclosure}, in this region, all entries in $\bm P$ are strictly positive. The following theorem presents the optimal mechanism in the high-privacy regime. Its proof can be found in Appendix \ref{app:highprivacyproof}.

\begin{theorem}
\label{thrm:highprivacymechanism}
Assume $X$ is distributed according to $P_X$. If $\epsilon \in [0, \epsilon_1(P_X))$, an instance of the optimal privacy mechanisms $[\bm P^*]$ in problem~\eqref{eq:optim_prob_cardinality} is 
\begin{equation}
\label{eq:optMIlowPML}
    P^*_{Y|X=x_i}(y_j) = \begin{cases}
        1 - e^{\epsilon}(1-P_X(x_i)) & \textnormal{if}\quad i=j, \\
        e^{\epsilon}P_X(x_j) & \textnormal{if}\quad i \neq j,
    \end{cases}
\end{equation}
where $i,j \in [N]$.     
\end{theorem}

\modified{
Note that for binary alphabets, this mechanism is a member of the equivalence class (see Remark \ref{remark:equivalenceclass}) of the binary mechanism~\eqref{eq:binarymech} in the case $P_X(x_1)\leq e^{-\epsilon}$. To illustrate the result of Theorem \ref{thrm:highprivacymechanism} for $N>2$, consider the following example. 
\begin{example}
    \label{ex:highprivacymech}
    Assume $P_X = (\nicefrac{2}{5},\nicefrac{1}{5},\nicefrac{1}{5},\nicefrac{1}{5})$ and $\epsilon = \log \nicefrac{9}{8}$. Then the optimal mechanism in Theorem \ref{thrm:highprivacymechanism} is
    \begin{equation}
       \bm P^* =  \begin{bmatrix}
            0.325 & 0.225 & 0.225 & 0.225 \\
            0.45 & 0.1 & 0.225 & 0.225 \\
            0.45 & 0.225 & 0.1 & 0.225 \\
            0.45 & 0.225 & 0.225 & 0.1
        \end{bmatrix}.
    \end{equation}
\end{example}
}

%% file: sections/closedformunifgeneral.tex
\label{sec:optsol:subsec:generalunif}

Next, we present a closed-form solution for optimal mechanisms in all privacy regions and a uniform distribution of the private data. The mechanisms are optimal for all \emph{permutation-symmteric} sub-convex functions. Note that many instances of sub-convex functions used in reality \modified{satisfy} this condition (mutual information, TV-distance, \dots). 
As discussed in Section \ref{sec:prelim:subsec:regions}, the privacy region of $\epsilon$ determines the maximum number of zero-valued elements in any column. For the case of uniform priors, it is irrelevant \emph{which} of the elements are set to zero. As a result, the optimization becomes independent of the realized symbol $y$, and columns of the optimal mechanism are permutations of one another, arranged such that a mechanism \modified{satisfies} the row-stochasticity constraint.

\begin{theorem}
\label{thrm:allregionseqprioropt}
Assume $\epsilon$ in some arbitrary privacy region $k \in [N-1]$. Suppose the private data $X$ is uniformly distributed. Then, assuming that $\mu$ is a \modified{permutation-}symmetric function,\footnote{\modified{A function $\mu: \mathbb R_+^N \to \mathbb R_+$ is said to be permutation-symmetric iff $\mu(\bm z) = \mu(\bm z \bm \Pi)$ for all $\bm z \in \mathbb R_+^N$, where $\bm \Pi$ is any arbitrary permutation matrix.}} \modified{an} optimal privacy mechanism in problem \eqref{eq:optim_prob_cardinality} is \modified{part of the equivalence class $[P^*_{Y|X}]$ with}
\begin{align}
\label{eq:allregionseqprioropt}
    &P_{Y|X=x_i}^*(y_j) =\\
    &\begin{cases}
    \frac{e^{\epsilon}}{N}, & \text{ if } i \in \{j+1,\dots,\text{mod}(j+(N-k),N)\}, \\
    1-e^{\epsilon}\frac{(N-k)}{N}, & \text{ if } \modified{i = j}, \\
    0, & \text{ o/w},
    \end{cases}\nonumber
\end{align}
where $i,j \in [N]$. That is, each column of the optimal mechanism has exactly $N-(k-1)$ non-zero elements, of which $N-k$ take the value $\frac{e^{\epsilon}}{N}$ and one has value $1-e^{\epsilon}\frac{(N-k)}{N}$.
\end{theorem}

The proof of this theorem is provided in Appendix \ref{app:allregionseqprioropt}. \modified{To illustrate the structure of the mechanisms in Theorem \ref{thrm:allregionseqprioropt}, we give the following example.

\begin{example}
    Assume $P_X = (\nicefrac{1}{4},\nicefrac{1}{4},\nicefrac{1}{4},\nicefrac{1}{4})$ and $\epsilon = \log 3$. Hence $\epsilon \in [\epsilon_2(P_X),\epsilon_3(P_X))$, that is, $\epsilon$ is in the $3^{\text{rd}}$ privacy region. An optimal mechanism according to Theorem \ref{thrm:allregionseqprioropt} is
    \begin{equation}
       \bm P^* = \begin{bmatrix}
            0.75 & 0.25 & 0 & 0 \\ 0 & 0.75 & 0.25 & 0 \\
            0 & 0 & 0.75 & 0.25 \\ 0.25 & 0 & 0 & 0.75 
        \end{bmatrix}.
    \end{equation}
\end{example}
}
The following result characterizes the mutual information when privatizing data using the presented mechanism with uniform priors.

\begin{corollary}
    Assume the private data $X$ to have alphabet size $N$ and be distributed according to a uniform prior distribution. Then the mutual information utility the mechanism in \eqref{eq:allregionseqprioropt} achieves is 
    \begin{equation}
    \label{eq:MIreductionUnif}
        I(X;Y) = \log N - H\biggl((\underbrace{\frac{e^{\epsilon}}{N},\dots,\frac{e^{\epsilon}}{N}}_{(N-k) \text{-times}},1-\frac{(N-k)e^{\epsilon}}{N})\biggr),
    \end{equation}
    where $H$ is the entropy function of appropriate dimension.
\end{corollary}

Note that with the mutual information utility and uniform priors, an $\epsilon_{\max}$-PML private mechanism, that is, a mechanism with $X = Y$ or no privatization, achieves $I(X;Y) = \log N$. Therefore, the entropy term in \eqref{eq:MIreductionUnif} can be seen as the privatization cost of the PML-optimal mechanism for some $\epsilon < \epsilon_{\max}$. For \modified{the smallest parameter} $\epsilon = 0$, this privatization cost attains the value $\log N$, which yields $I(X;Y) = 0$.  

%% file: sections/lp.tex
\label{sec:optsol:subsec:generalLP}

In this section, we present a \modified{reduced-complexity} linear program that \modified{improves upon the general vertex enumeration approach in Section \ref{sec:tradeoff:subsec:optregion} for computing} optimal mechanisms under the most general assumptions, that is, all privacy regions and arbitrary prior distributions. The presented linear program operates on a collection of extremal \emph{lift vectors}, the size and elements of which depend on the privacy parameter $\epsilon$, as well as the prior distribution $P_X$. \modified{For any fixed $y$, we define a lift vector $\bm \lambda(y)$ to be the posterior distribution vector of $X$ given $Y=y$, normalized by each corresponding prior probability of the symbols in $\mathcal X$, that is, 
\begin{equation}
    \bm \lambda(y) = \biggl(\frac{P_{X|Y=y}(x_1)}{P_X(x_1)},\dots,\frac{P_{X|Y=y}(x_N)}{P_X(x_N)} \biggr)^T.
\end{equation}
We then use the fact that, due to the homogeneity of sub-linear functions, for any fixed $y$ we have
\begin{equation}
    \mu(\bm P_y) = \mu(\bm \lambda(y) P_Y(y)) = P_Y(y) \mu(\bm \lambda(y)),
\end{equation}
and the constraint imposed on $\bm \lambda$ by the given PML requirement, i.e., $\max_x (P_{X|Y=y}(x)/P_X(x)) \leq e^{\epsilon}$, is independent of $y$. This allows us to split the optimization of $P_{Y|X}$ into separatly optimizing $P_{X|Y}$, and finding the optimal output distribution $P_Y$.}
Assume a fixed prior distribution $P_X$ and $\epsilon$ in the $k^{\text{th}}$ privacy region. Consider the polytope of feasible lift vectors \modified{as}
\begin{align}
\begin{split}
\label{eq:liftVertices}
\modified{\bm V(k,P_X) \coloneqq }\biggl\{\bm \lambda \in &[0,e^{\epsilon}]^N : \sum_{i = 1}^N \lambda_i P_X(x_i)=1, \\ &\text{ and } \sum_{i = 1}^N \mathbb I (\lambda_i > 0) P_X(x_i)\geq e^{-\epsilon}\biggr\},
\end{split}
\end{align}
where $\mathbb I(\cdot)$ denotes the indicator function. We denote the set of all vertices of this polytope by \modified{$\bm V^*(k,P_X)\coloneqq \{\bm \lambda^*_j\}$}. Due to the maximization of a convex function on a polytope, it can be shown that the columns of the optimal \modified{posterior distribution matrix} take values in the set \modified{$\bm V^*(k,P_X)$}. The linear program \modified{in Theorem \ref{thrm:LP} then} determines \emph{which} of \modified{these} vertices compose an optimal solution, and computes $P_Y(y_j)$ for $j \in [N]$ as the \say{weights} assigned to each of the selected vertices in the corresponding optimal solution.

Algorithm~\ref{alg:vertexenumeration} describes \modified{a} procedure for finding the elements of \modified{$\bm V^* (k,P_X)$. First, we initialize $\bm V^*$ to be the empty set (Line 1). Then, Line 2-9 consist of $k$ iterations of all privacy regions up to $k$. For each region, we iterate over all possible positions to have $N-l+1$ non-zero elements (Line 3).\footnote{\modified{We use $\binom{N}{N-l+1}$ to denote the set of all possible combinations of $N-l+1$ elements out of the set $[N]$.} For more details on algorithms generating combinations, see, e.g., \cite{payne1979combination}.} In Line 4 we check that the condition for a feasible point in \eqref{eq:liftVertices} is satisfied. If this is the case, we iterate over all of the indices in the set $J$ (Line 6). In each iteration, the current $j$ acts as the non-extremal element of the lift vector whose value is calculated in Line 8. All other indices in $J$ get the extremal value $e^\epsilon$ (Line 9). Each such configuration is then added to the set $\bm V^*$ if $\lambda_j \geq 0$. Then, having constructed the set $\bm V^*$, we obtain the following linear program for computing a $\bm P^* \in [P_{Y|X}^*]$.}

\begin{theorem}
\label{thrm:LP}
 Suppose $X$ is distributed according to $P_X$, and assume $\epsilon$ is in privacy region $k \in [N-1]$.  Then, assuming that $\mu$ is \modified{permutation-symmetric}, the optimal privacy mechanism in problem~\eqref{eq:optim_prob_cardinality} can be found by the linear program 
\begin{align}
\begin{split}
    \max_{P_Y(y_j), \, j \in [N]}\, &\sum_{j= 1}^{|\bm V^*(k,P_X)|}P_Y(y_j) \mu(\bm \lambda^*_j) \\
    \text{s.t.} \quad &\sum_{j= 1}^{|\bm V^*(k,P_X)|} P_Y(y_j) = 1 \\
         \quad &\sum_{j= 1}^{|\bm V^*(k,P_X)|}P_Y(y_j) \lambda^*_{ij} = 1 \quad \forall i \in [N], 
\end{split}         
\end{align}
where $\lambda^*_{ij}$ \modified{denotes the $i^{\text{th}}$ element of the vertex} $\bm \lambda^*_j$.
\end{theorem}

\modified{The proof of Theorem \ref{thrm:LP} is given in Appendix \ref{app:LP}. We remark that the linear program in Theorem \ref{thrm:LP} has a significantly lower complexity than the vertex enumeration approach in Section \ref{sec:tradeoff:subsec:optregion}. Especially in scenarios in which $k \ll N$, the cardinality of the set $\bm V^*(k,P_X)$ stays comparatively small (see Remark \ref{remark:LPsize}). Further, the separation of $P_{Y|X}$ into the $y$-independent lift vectors $\bm \lambda \in \bm V^*(k,P_X)$ and \say{weights} $P_Y(y)$ significantly reduces the dimensionality of the linear program.}

\begin{remark}
\label{remark:LPsize}
    Evidently, the size of the set \modified{$\bm V^*(k,P_X)$}, and therefore the computational complexity of the linear program in Theorem \ref{thrm:LP}, grows with increasing value of $k = 1,\dots,N-1$. More precisely, we can upper bound the number of extremal lift vectors in the set for a fixed value of $k$ as 
    \begin{equation}
    \label{eq:LPorderupperbound}
        |\bm V^*(k,P_X)| \leq \sum_{l=1}^k (N-l+1)\binom{N}{l-1},
    \end{equation}
    and \eqref{eq:LPorderupperbound} holds with equality if $P_X$ is \modified{the} uniform distribution.
\end{remark}

\modified{
\begin{algorithm2e}[t]
\setstretch{1.125}
\DontPrintSemicolon
\caption{Vertex enumeration of $\bm V (k,P_X)$}\label{alg:vertexenumeration}
\SetKwInOut{Input}{input}\SetKwInOut{Output}{output}
\Input{Privacy parameter $\epsilon$ in privacy region $k$, prior distribution $P_X$, input alphabet size $N$.}
\Output{$\bm V^* (k,P_X)$}
\nl Initialize $\bm V^* \gets \emptyset$; \\
\nl \For{$l \in [k]$}
    {
    \nl \For{each set of $N-l+1$ indices $J \in \binom{N}{N-l+1}$}  {
          \nl \If{$e^{-\epsilon} \leq \sum_{m \in J}P_X(x_m)$}
          {
            \nl Initialize $\bm \lambda = {0}^N$\;
            \nl \For{each index $j \in J$}
            {    
            \nl \If{$e^{\epsilon}\sum_{m \in J\setminus\{j\}}P_X(x_m) \leq 1$}
            {
                \nl Set $\lambda_j = \frac{1-e^{\epsilon}\sum_{m \in J\setminus\{j\}}P_X(x_m)}{P_X(x_j)}$\;
                \nl $\forall i \in J\setminus\{j\}$ set $\lambda_i = e^{\epsilon}$\;
                \nl Add $\bm \lambda$ to $\bm V^*$\;
                }
            }
            }
           }
    }
    \nl \Return{$\bm V^*$}\;
\end{algorithm2e}
}
To illustrate the structure of the set $\bm V^*(k,P_X)$, consider the following example.

\begin{example}
    Assume some arbitrary ternary random variable $X$ distributed according to $P_X$. Assume further that $\epsilon$ is in the second privacy region, that is, the privacy budget allows for one zero in the feasible lift vectors. Then, the set $\bm V^*(2,P_X)$ has the structure
    \begin{equation}
    \left\{ \begin{pmatrix}
        e^{\epsilon} \\
        0 \\
        r_3(1)
    \end{pmatrix}, \begin{pmatrix}
         r_1(2)  \\
         e^{\epsilon} \\
         \modified{0}
    \end{pmatrix}, \begin{pmatrix}
        e^{\epsilon} \\
        r_2(1) \\
        0
    \end{pmatrix}, \dots, \begin{pmatrix}
        e^{\epsilon} \\
        e^{\epsilon} \\
        r_3(3)
    \end{pmatrix} , 
    \right\}
    \end{equation}
    where $r_{i}(j) \coloneqq \frac{1-e^{\epsilon}P_X(x_j)}{P_X(x_i)}$ \modified{and $r_3(3) \coloneqq \frac{1-e^{\epsilon}(1-P_X(x_3))}{P_X(x_3)}$}.
\end{example}

%% file: sections/numericalResults.tex
\label{sec:optsol:numerical}
\modified{In this section, we illustrate some of the derived mechanisms using numerical examples. We use synthetic data with a specified prior to model $n$ instances of a secret $X$, labeled as the sequence $X^n \coloneqq (X_1,\dots,X_n)$. For each $\epsilon \in [0,\epsilon_{s}]$ spaced equidistant with $\Delta \epsilon$, we then privatize each of the symbols $X_i$, creating the privatized sequence $Y^n = (Y_1,\dots,Y_n)$ using an instance of the presented mechanisms. Here, the pair $(\epsilon_s,\Delta \epsilon)$ is chosen as either $(\epsilon_{\max},0.005)$ (binary mechanism) or $(\epsilon_1(P_X),0.0005)$ (high-privacy mechanism). For each fixed $\epsilon$, we calculate the corresponding privacy parameter $\epsilon_r$ for the randomized response mechanism from~\eqref{eq:epsldpofepspml} such that both mechanisms achieve the same PML. We repeat each experiment 10 times with a sample size $n=1000$.

\subsubsection{Mutual information}
We evaluate mechanism performance using the empirical estimate of the mutual information
\begin{equation}
   \hat{I}(X^n;Y^n) = \sum_{i=1}^N \sum_{j=1}^N \frac{f(x_i,y_j)}{N}\log\frac{Nf(x_i,y_j)}{f(x_i)f(y_j)},
\end{equation}
where $f(x_i,y_j)$ denotes the frequency of the tuple $(x_i,y_j)$ in the sequence $(X^n,Y^n)\coloneqq \{(X_1,Y_1),\dots,(X_n,Y_n)\}$, and $f(x_i)$ and $f(y_j)$ denote the frequencies of symbols $x_i$ and $y_j$ in the sequences $X^n$ and $Y^n$, respectively.   
In Figure \ref{fig:empiricalutilitybin}, we compare the empirical mutual information of the optimal binary mechanism presented in Section \ref{sec:optsol:subsec:binary}, as well as the optimal high-privacy mechanism in Section \ref{sec:optsol:subsec:highprivacy} to the corresponding randomized response mechanism. There is a clear increase in utility for the mechanisms optimal for PML.\footnote{\modified{Note that for a binary uniform prior distribution, the two mechanisms are identical. In fact, in this specific case, PML and LDP are equivalent privacy measures, see \cite[Example 1]{grosse2024quantifying}.}}

\subsubsection{Pearson correlation coefficient} 
In Figure \ref{fig:corrcoef}, we further demonstrate mechanism performance under the empirical Pearson correlation coefficient $r(X^n;Y^n)$ defined as
\begin{equation}
r(X^n;Y^n) = \frac{\sum_{i=1}^{n} (X_i - \bar{X})(Y_i - \bar{Y})}{\sqrt{\sum_{i=1}^{n} (X_i - \bar{X})^2 \sum_{i=1}^{n} (Y_i - \bar{Y})^2}},
\end{equation}
with $\Bar{X}$, $\Bar{Y}$ denoting the sample mean of $X^n$ and $Y^n$, respectively and each $X_i, Y_i \in [N]$. We remark that the Pearson correlation coefficient is not a sub-convex utility function. Hence, the results in Figure \ref{fig:corrcoef} demonstrate that the presented mechanisms are able to increase utility compared to randomized response even for utility functions that do not satisfy sub-convexity.
}
\begin{figure}[!t]
    \centering
    \begin{subfigure}[b]{.99\columnwidth}
    \centering
            \includegraphics[scale=0.5]{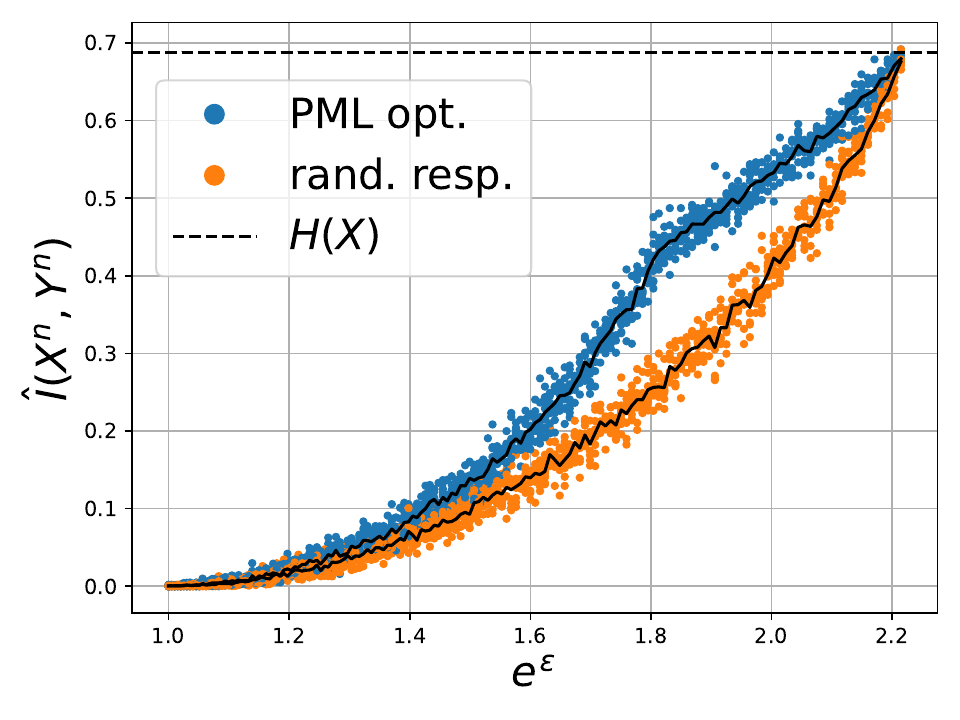}
            \caption{Binary mechanism, $P_X = (0.55,0.45)$}
    \end{subfigure}
    \par\bigskip
    \begin{subfigure}[b]{.99\columnwidth}
        \centering
            \includegraphics[scale=0.5]{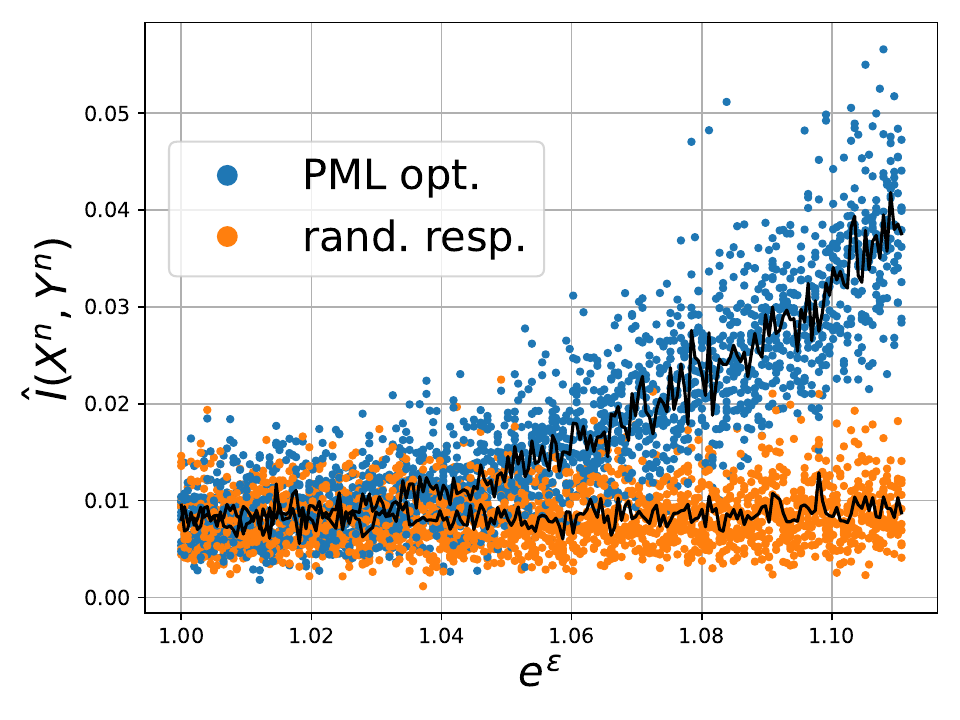}
            \caption{High-privacy mechanism, $P_X = (0.3,0.2,0.2,0.2,0.1)$}
    \end{subfigure}
    \caption{\modified{Empirical mutual information of the optimal binary mechanism~\eqref{eq:binarymech} and the optimal high-privacy mechanism~\eqref{eq:RRpml} for non-uniform prior distributions, and how they compare to the randomized response mechanism in~\eqref{eq:LDP-RR}. The black line indicates the mean value of the $10$ experiments.}}
    \label{fig:empiricalutilitybin}
\end{figure}

\begin{figure}[!t]
    \centering
    \begin{subfigure}[b]{.99\columnwidth}
    \centering
        \includegraphics[scale=0.5]{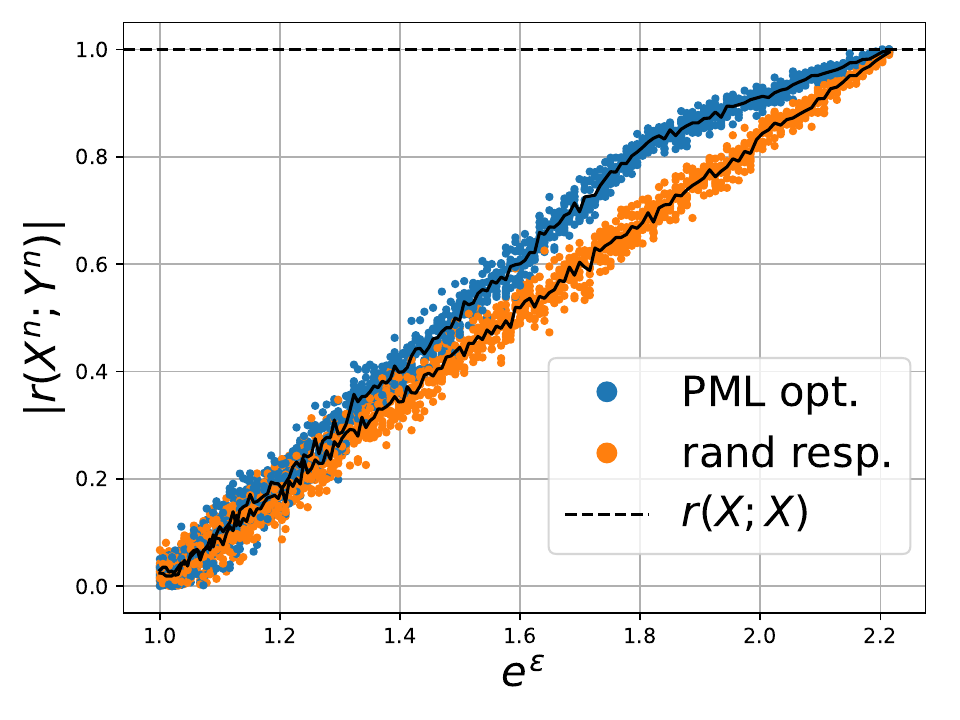}
        \caption{Binary mechanism, $P_X = (0.55,0.45)$}
        \label{fig:corrcoef:subfig:N2}
    \end{subfigure}
    \par\bigskip
    \begin{subfigure}[b]{.99\columnwidth}
    \centering
        \includegraphics[scale=0.5]{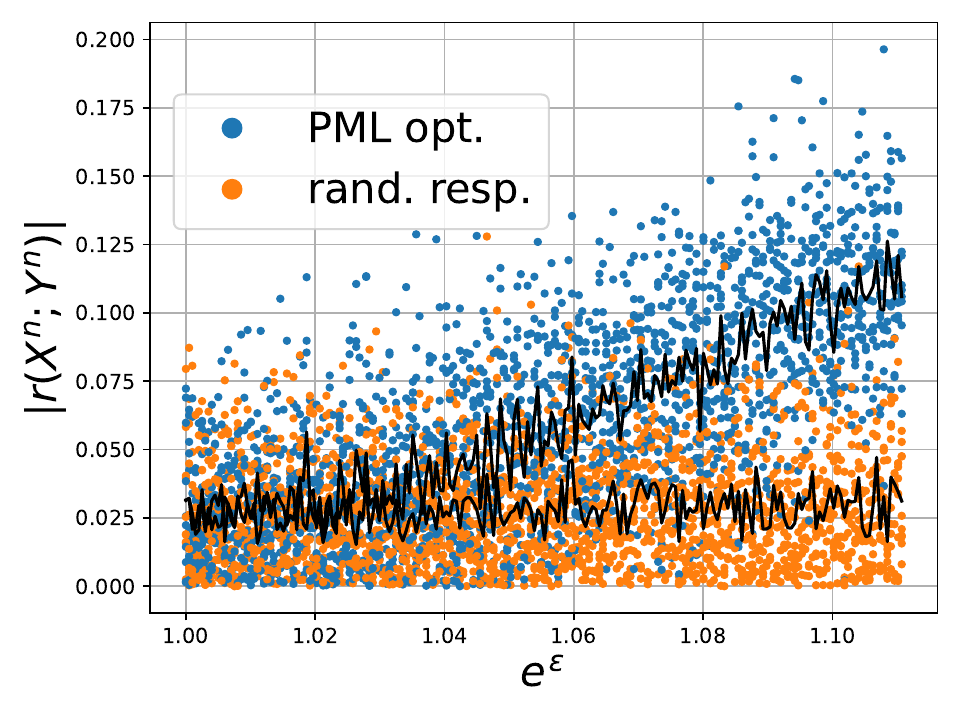}
        \caption{High-privacy mechanism, $P_X = (0.3,0.2,0.2,0.2,0.1)$}
        \label{fig:corrcoef:subfig:N5}
    \end{subfigure}
    \caption{\modified{Emprical correlation coefficient of the optimal binary and optimal high-privacy mechanism compared to the randomized response mechanism in \eqref{eq:LDP-RR}. The black line indicates the mean value of the $10$ experiments.}}
    \label{fig:corrcoef}
\end{figure}

%% file: sections/conclusions.tex
\label{sec:conclusions}
In this paper, we have used the PML framework to \modified{analyze} a general privacy-utility tradeoff problem that is central to the privacy-by-design approach as it is mandated by, e.g., \modified{the European General Data Protection Regulation (}GDPR\modified{) \cite{EuropeanParliament2016a}}. We presented various closed-form optimal solutions to the mechanism design problem with PML and sub-convex utility. We \modified{also} showed that computing optimal mechanisms even in the most general case is tractable and can be done via a linear program. Further, we demonstrated how the inferential results of PML regarding disclosure can be used to directly guide mechanism design by defining a set of privacy regions dependent on the privacy budget that directly affect the structure of the mechanism matrix. The results constitute an important first step in implementing PML privacy, and are essential for designing more complex data processing systems with strict PML guarantees at a minimal loss of performance.

%% file: sections/appendix.tex
\subsection{Proof of Theorem \ref{thrm:cardinality}}
\label{app:proofOfCardinality}
We follow the \emph{perturbation method}~\modified{\cite{gohari2012evaluation}}, \cite[Appendix C]{el2011networkinformationtheory}. Fix some $M > N$ and consider the maximization problem
\begin{equation*}
    \max_{P_{Y|X} \in \left\{\bm P \in  \mathcal S_{N,M} \colon \epsilon_m(\bm P) \leq \epsilon \right\},} \, U(P_{Y|X}).
\end{equation*}
Let $P^*_{Y \mid X}$ and $\bm P^*$ both denote the mechanism that achieves the maximum. Let $P^*_{XY} = P^*_{Y \mid X} \times P_X$ denote the optimal joint distribution. Given a constant $\gamma \in \mathbb R$ and a mapping $\phi : \mathcal Y \to \mathbb R$, let  $P_{XY}^{\gamma}$ be the perturbed version of $P^*_{XY}$ defined as $P_{XY}^{\gamma}(x,y) \coloneqq P^*_{XY}(x,y)(1+ \gamma \phi(y))$ for all $(x,y) \in \mathcal X \times \mathcal Y$. We assume that $1+ \gamma \phi(y) \geq 0$ for all $y \in \mathcal{Y}$ and $\mathbb{E}_{Y \sim P^*_{Y\mid X=x}}[\phi(Y)] = 0$ for all $x \in \mathcal{X}$. Then, the marginal distribution $P_X^\gamma$ is 
\begin{equation}
\label{eq:boundingcond1}
\begin{aligned}
    P_X^{\gamma}(x) &= \sum_{y \in \mathcal{Y}}P^*_{XY}(x,y)(1+ \gamma \phi(y)) \\&= \sum_{y \in \mathcal{Y}} P^*_{XY}(x,y) = P_{X}(x),
\end{aligned}
\end{equation}
for all $x \in \mathcal{X}$, implying that the distribution of $X$ is unaffected by the perturbation. Thus, we may write $P_{XY}^\gamma = P^\gamma_{Y \mid X} \times P_X$. Also, note that since the constraints $\mathbb{E}_{Y \sim P^*_{Y \mid X=x}}[\phi(Y)] = 0$ with $x \in \mathcal X$ specify at most $N$ linearly independent equations, a non-zero $\phi$ exists as long as $M > N$.

Now, by definition, mechanism $P^*_{Y \mid X}$ satisfies $\epsilon$-PML. We argue that the
perturbed mechanism $P^\gamma_{Y \mid X}$ also satisfies $\epsilon$-PML. To see why, note that $P^{\gamma}_{XY}(x,y) = P_{X}(x)P^*_{Y|X=x}(y)(1+\gamma \phi(y))$, which yields
\begin{equation}
    P^*_{Y|X=x}(y) = \frac{P^{\gamma}_{XY}(x,y)}{P_X(x)(1+\phi(y))},
\end{equation}
for all $x \in \mathcal{X}$ and $y\in\mathcal{Y}$. Therefore, we have
\begin{align*}
\ell_{P^*_{XY}}(X \to y) &= \max_x \frac{P^*_{Y|X=x}(y)}{\sum_x P^*_{Y|X=x}(y) \cdot P_X(x)}\\
&= \left(\underset{x}{\max}\,\frac{P^{\gamma}_{XY}(x,y)}{P_{X}(x)} \right) \cdot \frac{1}{\sum_x P^{\gamma}_{XY}(x,y)} \\[0.5em]
&= \frac{\underset{x}{\max}\,P^\gamma_{Y|X=x}(y)}{P^\gamma_{Y}(y)} \leq e^{\epsilon},
\end{align*}
for all $y \in \mathcal Y$, where $P^\gamma_Y$ denotes the marginal of $P^\gamma_{XY}$ over $Y$. That is, the perturbed distribution also yields a valid $\epsilon$-PML mechanism.

Next, we examine the utility of the perturbed mechanism:
\begin{equation}
\begin{aligned}
    U(P^\gamma_{Y|X}) &=\sum_{y \in \mathcal{Y}} \mu\biggl(\bm P^*_y (1+\gamma \phi(y))\biggr)\\
    &= \sum_{y \in \mathcal{Y}}(1+\gamma \phi(y)) \; \mu\biggl(\bm P^*_y \biggr)\\
    &= \sum_{y \in \mathcal{Y}} \; \mu\biggl(\bm P^*_y \biggr) + \gamma \sum_{y \in \mathcal{Y}}\phi(y) \; \mu \biggl(\bm P^*_y\biggr) \\
    &= U(P^*_{Y|X}) + \gamma \, U_{\phi}(P^*_{Y|X}),
\end{aligned}
\end{equation}
where $\bm P^*_y$ denotes the column in matrix $\bm P^*$ corresponding to outcome $y$, and $U_{\phi}(P^*_{Y|X}) \coloneqq \sum_{y \in \mathcal{Y}}\phi(y) \; \mu\biggl(\bm P^*_y \biggr)$. 
Since $P^*_{Y \mid X}$ is the optimal mechanism it must hold that
\begin{equation}
\label{eq:partialepsequalszero}
    \frac{\partial}{\partial \gamma}U(P^\gamma_{Y|X}) = U_{\phi}(P^*_{Y|X}) = 0,
\end{equation} 
and therefore, $U(P^\gamma_{Y|X}) = U(P^*_{Y|X})$. That is, the perturbed mechanism $P^\gamma_{Y|X}$ achieves the same utility as the optimal mechanism $P^*_{Y|X}$. 

Finally, we choose $\gamma$ to be the largest value such that $1 +  \gamma \phi(y)\geq 0$ holds for all $y \in \mathcal Y$. At this value of $\gamma$, there exists $y^* \in \mathcal Y$ satisfying $1 + \gamma \phi(y^*) = 0$, and consequently, $P^{\gamma}_Y(y^*) = 0$. This implies that the size of the support set of $Y$ can be reduced to $M-1$ while maintaining the $\epsilon$-PML privacy guarantee and without any loss to the utility. Furthermore, the above argument can be repeated as long as $M > N$. We conclude that, without loss of generality, we can restrict the feasible set of privacy mechanisms in problem~\eqref{eq:generalOPTproblem} to those with $N = M$, as desired.  \qed

\subsection{Proof of Lemma~\ref{lemma:polytope}}
\label{app:polyproof}
Fix $j \in [N]$. 
The privacy constraint for outcome $y_j$ can be expressed as 
\begin{equation}
    \exp\biggl(\ell(X \rightarrow y_j)\biggr) = \frac{\underset{i}{\max}\,\{p_{ij}\}}{\sum_i p_{ij}P_X(x_i)} \leq \e^{\epsilon}.
\end{equation}
Denoting the prior by $\boldsymbol{\pi} \coloneqq (P_X(x_1),\dots,P_X(x_N))^T \in [0,1]^N$, the above constraint can be rewritten as $(\boldsymbol{\pi}^T \bm P_j) e^{\epsilon} \geq p_{ij}$ for all $i \in [N]$. That is, for each $i \in [N]$ the constraint describes a closed half-space
\begin{equation}
    \{\bm P_j \geq 0: (\boldsymbol{\pi}^T \bm P_j) e^{\epsilon} - p_{ij} \geq 0\} 
    .
\end{equation}
The set $\mathcal{M}_N(\epsilon)$ consists of the intersection of these closed half-spaces for all $i \in [N]$, as well as the (closed) polytope imposed by the box constraints $0 \leq p_{ij} \leq 1$ (a unit $N^2$-hypercube) and the requirement on row-stochasticity, which defines a hyperplane. The intersection of finitely many closed half-spaces and polytopes is again a closed polytope \cite{boyd2004convex}. \qed

\subsection{Proof of Lemma \ref{lem:privacyregiondisclosure}}
\label{app:proofLemLowPrivacy}
Suppose $P_{Y \mid X}$ has \modified{a column with} $2 \leq l \leq N$ non-zero elements. We show that $\max_{y \in \mathcal{Y}}\ell(X\to y) \geq \epsilon_{N-l}(P_X) = -\log \sum_{i=1}^{l} P_X(x_i)$. By assumption, there exists some $y \in \mathcal{Y}$ and some $x_{i_1},\dots, x_{i_l} \in \mathcal {X}$ such that $P_{Y|X=x}(y) = 0$ for all $x \in \mathcal{X}\backslash\{x_{i_1},\dots,x_{i_l}\}$. Define the distribution $Q_X$ as $Q_X(x) = \frac{P_X(x)}{\sum_{j=1}^l P_X(x_{i_j})}$ for $x \in \{x_{i_1},\dots,x_{i_l}\}$ and assign $Q_X(x)=0$ for all other elements of $\mathcal{X}$. Note that $Q_X$ is a probability distribution on $\mathcal{X}$ with the support set $\{ x_{i_1},\dots,x_{i_l} \}$. Looking at the leakage of the mechanism, we have
\begin{equation}
\begin{aligned}
    &\ell_{P_{Y|X}\times P_X}(X\rightarrow y) = \log \frac{\max_{x\in \mathcal{X}}P_{Y|X=x}(y)}{P_Y(y)} \\
    &= \log \frac{\max_{x \in \{x_{i_1},\dots,x_{i_l}\}}P_{Y|X=x}(y)}{\sum_{x \in \{x_{i_1},\dots,x_{i_l}\}}P_{Y|X=x}(y)P_X(x)} \\
    &= \log \frac{\max_{x \in \{x_{i_1},\dots,x_{i_k}\}}P_{Y|X=x}(y)}{\Big(\sum_{j=1}^l P_X(x_{i_j}) \Big) \sum_{x \in \{x_{i_1},\dots,x_{i_k}\}}P_{Y|X=x}(y)Q_X(x)} \\
    &= \log \frac{1}{\sum_{j=1}^l P_X(x_{i_j})} + \ell_{P_{Y|X}\times Q_X}(X \to y) \\
    &\geq \log \frac{1}{\sum_{j=1}^l P_X(x_{i_j})} \geq \log \frac{1}{\sum_{i=1}^l P_X(x_{l})} = \epsilon_{N-l},
\end{aligned}
\end{equation}
where the first inequality is due to the non-negativity of PML \cite[Lemma 1]{pml}. \qed

\subsection{Proof of Theorem~\ref{thrm:binarymechanism}}
\label{app:binaryproof}
\modified{We prove this theorem by straightforwardly calculating the vertices of the polytope defined in Lemma \ref{lemma:polytope}. In the present binary case, this is equivalent to finding the intersecting point of the lines defined by the PML constraints. The solution approach is illustrated in Figure \ref{fig:BIBO-mi}.}
Let $\pi_1 \coloneqq P_X(x_1)$ and $\pi_2 \coloneqq P_X(x_2)$. Recall that $p_{ij} = P_{Y|X=x_i}(y_j)$. In the binary case, due to the row-stochasticity relation $p_{i2} = 1 - p_{i1}$ for $i=1,2$, a mechanism $\bm P$ is fully determined by its first column $\bm P_1 = (p_{11}, p_{21})^T$. Further, by the assumption \modified{that the prior probabilities are in non-increasing order}, we \modified{have} $\pi_1 \geq \pi_2$.

Now, note that $p_{11},p_{21} \in [0,1]$, and we can split the region $[0,1]^2$ into the disjoint sets $S_1 \coloneqq \{(p_{11},p_{21}) \in [0,1]^2: p_{11} \geq p_{21}\}$ and $S_2 \coloneqq [0,1]^2 \setminus S_1$. Considering the PML constraints on the first column $\bm P_1$ of a generic binary mechanism, we get the following boundaries on the optimization region.
\begin{equation}
\label{eq:boundary11}
    B_1^{y_1} \coloneqq \{(p_{11},p_{21}) \in S_1: p_{21} = \biggl(\frac{1-\pi_1e^{\epsilon}}{\pi_2 e^{\epsilon}}\biggr) p_{11}\}
\end{equation}
and
\begin{equation}
\label{eq:boundary21}
     B_2^{y_1} \coloneqq \{(p_{11},p_{21}) \in S_2: p_{11} = \biggl(\frac{1-\pi_2e^{\epsilon}}{\pi_1 e^{\epsilon}}\biggr) p_{21}\}.
\end{equation}
Similarly, for the second column of the mechanism, that is, for $Y=y_2$, \modified{by using $p_{i2} = 1-p_{i1}$} we get
\begin{equation}
\label{eq:boundary12}
     B_1^{y_2} \coloneqq \{(p_{11},p_{21}) \in S_1: p_{21} = \biggl(\frac{\pi_1e^{\epsilon}}{1- \pi_2 e^{\epsilon}}\biggr) p_{11} + \frac{1-e^{\epsilon}}{1-\pi_2e^{\epsilon}}\}
\end{equation}
and
\begin{equation}
\label{eq:boundary22}
     B_2^{y_2} \coloneqq \{(p_{11},p_{21}) \in S_2: p_{11} = \biggl(\frac{\pi_2e^{\epsilon}}{1- \pi_1 e^{\epsilon}}\biggr) p_{21} + \frac{1-e^{\epsilon}}{1-\pi_1e^{\epsilon}}\}.
\end{equation}
The above sets describe the boundaries of $\mathcal M_2(\epsilon)$. Hence, their intersections yield \modified{the desired} extremal mechanisms. 

Fig. \ref{fig:BIBO-mi} depicts $\mathcal M_2(\epsilon)$ for different priors. The contour lines in the figure illustrate the value of mutual information between $X$ and $Y$. For $p_{11}=p_{21}$ mutual information attains its minimum value of zero. Therefore, the points $(p_{11},p_{21}) = (0,0)$ and $(p_{11},p_{21}) = (1,1)$, which yield two of the four extremal mechanisms in the binary case, can be disregarded. As they result in zero utility\modified{, they can be seen as trivial solutions}. The search for a maximizing vertex can therefore be limited to the two intersections of the linear constraints that lie strictly inside of $S_1$ and $S_2$. Let the optimal non-trivial solution in each region $S_i$ be denoted as $\bm P_1^{(*,i)} \coloneqq (p^*_{11},p^*_{21})^T$, $i = 1,2$. 
Observe that $\pi_1 e^{\epsilon} \geq 1$ enforces the constraint $0 \leq p_{21} \leq 1$ to be active, therefore implying $p_{21}^{*} = 0$ for the boundary $B_1$ and $p_{21}^{*} = 1$ for $B_2$. 

Solving equations \eqref{eq:boundary11} - \eqref{eq:boundary22} when $\pi_1 e^\epsilon \leq 1$ and \modified{consequently} $\pi_2 e^\epsilon \leq 1$ yields
\begin{equation}
    \bm P_1^{(*,2)} = \begin{bmatrix}\,\pi_2 e^{\epsilon} \\ 1-\pi_1 e^{\epsilon}\,\end{bmatrix}
\end{equation}
\begin{equation}
    \bm P_1^{(*,1)} = \begin{bmatrix}1-\pi_2 e^{\epsilon} \\ \pi_1 e^{\epsilon}\end{bmatrix} = 1 - \bm P_1^{(*,2)}.
\end{equation}

\modified{On the other hand,} for the case \modified{$\pi_1e^{\epsilon} \geq 1$, that is whenever }$p_{21} \in \{0,1\}$, we have
\begin{equation}
    \bm P_1^{(*,1)} = \begin{bmatrix}
        \frac{e^{\epsilon}-1}{\pi_1 e^{\epsilon}}, & 0\end{bmatrix}^T, \quad
    \bm P_1^{(*,2)} = \begin{bmatrix}\frac{1-\pi_2 e^{\epsilon}}{\pi_1 e^{\epsilon}}, & 1\end{bmatrix}^T,
\end{equation}


Further, due to the invariance of the utility value to column permutations of the mechanism, the two mechanisms associated with the extreme points $\bm P_1^{(*,1)}$ and $\bm P_1^{(*,2)}$ are members of the same equivalence class $[\bm P^*]$ according to Remark \ref{remark:equivalenceclass}. Therefore, both $\bm P_1^{(*,1)}$ and $\bm P_1^{(*,2)}$ achieve the optimal utility value $U^*(\epsilon)$. \modified{Noticing that these solutions result in the mechanism matrix~\eqref{eq:binarymech} proves the theorem.} \qed

\begin{figure}[!t]
\centering
\begin{subfigure}[b]{.99\columnwidth}
    \centering
    \includegraphics[scale=0.8]{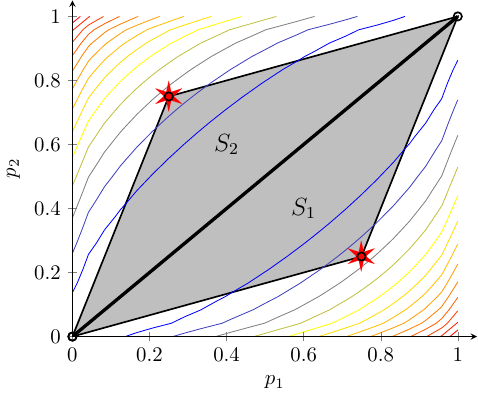}
    \caption{$P_X(x_1) = P_X(x_2)$}
    \label{fig:BIBO-mi:subfig:sym}
\end{subfigure}
\par\bigskip
\begin{subfigure}[b]{.99\columnwidth}
    \centering
    \includegraphics[scale=0.8]{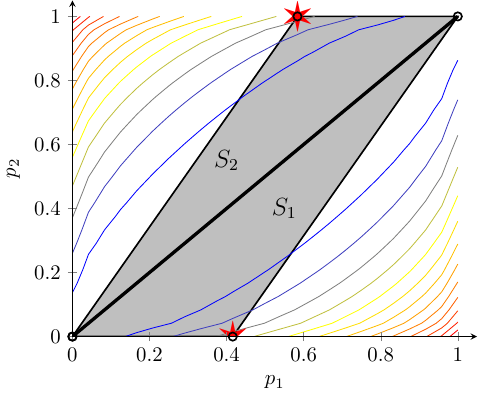}
    \caption{$P_X(x_1) = 0.8$.}
    \label{fig:BIBO-mi:subfig:asym}
\end{subfigure}
\caption[Geometry of the binary optimization problem for different prior distributions]{Geometry of the binary input binary output optimization problem for different prior configuration and $\epsilon = \log 1.5$. Contour lines show mutual information between source alphabet and the variable induced by the mechanism. Stars mark optimal solutions. Note that \ref{fig:BIBO-mi:subfig:asym} corresponds to a case where $e^{\epsilon} P_X(x_1) \geq 1$, that is, a case with an active box constraint. \ref{fig:BIBO-mi:subfig:sym} corresponds to a case with no active box constraints. The dividing line indicates the regions $S_1$,$S_2$.}
\label{fig:BIBO-mi}
\end{figure}

\subsection{Proof of Theorem \ref{thrm:highprivacymechanism}}
\label{app:highprivacyproof}
As discussed in Section \ref{sec:optsol:subsec:cardinality}, we will consider the mechanism $\bm P = P_{Y|X}$ to be a $N \times N$ matrix. We say that a mechanism $\bm Q$ has a smaller output \modified{support} size than $\bm P$ if $\bm Q$ has more all-zero columns than $\bm P$, implying $\supp(Q_{Y|X}) < \supp(P_{Y|X})$.

For simplicity, each step in the proof is given as a separate lemma. We start by establishing the overall structure of the extremal mechanisms. Since extremal mechanisms are the extreme points of the polytope $\mathcal M_N(\epsilon)$, they satisfy $N^2$ of the constraints in \eqref{eq:constraints} with equality. Note that all privacy mechanisms (extremal or not) satisfy the $N$ equality constraints of \eqref{subeq:sumconstraint}. Therefore, the distinguishing factor for extremal mechanisms is that they also satisfy $N(N-1)$ of the inequality constraints in \eqref{subeq:pmlconstraints} or \eqref{subeq:nonnegativityconstraint} with equality. This fact is frequently used in the proof.

Throughout this section, we will refer to \emph{extremal elements} as the elements of a mechanism that meet a PML constraint with equality. \modified{That is,} we call an element $p_{ij}$ of a mechanism $P_{Y|X}$ extremal if $\frac{p_{ij}}{P_Y(y_j)} = e^{\epsilon}$, where $i,j \in [N]$.    
\modified{We start by giving an overview over the proof's ideas. 

As mentioned above, in order to qualify for a maximizing solution, a mechanism needs to be a vertex of the optimization region polytope. It is easy to check that the mechanism in Theorem \ref{thrm:highprivacymechanism} meets this requirement, as it satisfies $N(N-1)$ inequality constraints with equality. What remains to show is that all \emph{other} mechanisms that meet this requirement achieve equal or lower utility.

To prove this, we first establish the structure of the vertices with output support size $N$ in Lemma \ref{lem:optimalstructure} and of mechanisms with output support size $M<N$ in Lemma \ref{lem:loweralphabetoptimalstructure}. Then, in Lemma \ref{lem:mergeisoptimal}, we show that merging two columns of an extremal mechanism results in another extremal mechanism with its output support size decreased by one. In Lemma \ref{lem:mergeexactlyonecolumn}, we use this fact to show that any extremal mechanism with output support size $N-k$ can be obtained form another extremal mechanism with output support size $N-k+1$ by such a merging operation. Finally, Lemma \ref{lem:nooptimalzerocolumns} uses the data processing inequality for sub-convex function to conclude that any extremal mechanism with lower output support size cannot achieve increased utility, as it can be obtained by successively merging the columns of an extremal mechanism with output support size $N$, an operation which cannot increase utility.}

\begin{lemma}
\label{lem:optimalstructure}
    If $P_{Y|X} \in \mathcal M_N(\epsilon)$ is extremal \modified{and has full output support}, then each $y_j$ fulfills exactly $N-1$ PML constraints with equality and the index $i \in [N]$ of the non-extremal element $p_{ij}$ is different for every $j \in [N]$.
\end{lemma}
\begin{proof}
    Firstly, assume $\epsilon > 0$, since otherwise there is nothing to prove and we have the trivial optimal solution of $X \indep Y$, that is, $P_{Y|X=x}(y) = \text{const}$. We prove this lemma by contradiction: For each outcome $j \in [N]$, we use $r(j) \in [N]$ to denote the index of a non-extremal element. We prove that all elements in column $j$ of the mechanism are equal except for the $r(j)$-th element. That is, 
    \begin{equation}
        p_{ij} = p_{i'j}, \quad \forall i,i' \in [N] \text{ and } i,i' \neq r(j).  
    \end{equation}
    We also show that $r(j) \neq r(j') \text{ for } j\neq j'$. 
    
    Note that as a result of Lemma \ref{lem:privacyregiondisclosure}, mechanisms in the high-privacy region $\epsilon < \log \frac{1}{1-p_{\min}}$ cannot contain elements equal to zero or one (except if we have an all-zeros column). Therefore, when $\supp(P_Y) = \supp(P_X)$, a mechanism can only be extremal by fulfilling exactly $N(N-1)$ of the PML constraints \eqref{subeq:pmlconstraints} with equality. Now, assume  there exists some outcome $y_j$ that fulfills $N-k$ of the PML constraints with $k \geq 2$. Then, in order to meet $N(N-1)$ privacy constraints, there must be $(k-1)$ outcomes each satisfying $N$ privacy constraints with equality. Suppose \modified{w.l.o.g. that} $y_1$ is such an outcome. Then, $p_{i1} = p_{i'1}$ for all $i,i' \in [N]$ and 
    \begin{equation}
    \label{eq:allNPMLcontradiction}
        \ell(X\rightarrow y_1) = \log\biggl(\frac{p_{11}}{p_{11}\sum_x P_X(x)}\biggr) = 0,
    \end{equation}
    i.e., $y_1$ meets no PML constraints, which is a contradiction. This shows that each outcome of the privacy mechanism must satisfy exactly $N-1$ PML constraints.  
    
    What is left to show is that $r(j) \neq r(j') \text{ for } j\neq j'$ holds. Again, we can construct a simple contradiction: Assume there exists $j \neq j'$ such that $r(j) = r(j')$. Since each column needs to fulfill $N-1$ PML constraints, there will be one row for which all elements are extremal, resulting in a violation of row-stochasticity. For this specific row, say at index $i \in [N]$, \modified{due to the assumption that $\epsilon > 0$} we get 
    \begin{equation}
        \sum_{j} p_{ij} = \sum_{j} e^{\epsilon} P_Y(y_j) = e^{\epsilon} > 1, 
    \end{equation}
    which yields a contradiction, as desired. 
\end{proof}

\modified{Lemma \ref{lem:optimalstructure} shows that} for each $j \in [N]$, $N-1$ of the elements in that column take the same value $p_{ij} = E(j)$, with $E(j)$ denoting the extremal value in column $j$. The remaining element $p_{r(j)j}$ takes the value ensuring that the row-sum constraints are met in each row. There are $N!$ such combinations, that all belong to the same equivalence class according to Remark \ref{remark:equivalenceclass}, that is, they are \modified{identical} up to column permutations. To illustrate this, consider $N=3$. Two of the possible solutions are
\begin{equation}
    \begin{bmatrix}
        p_{r(1)1} & E(2) & E(3) \\
        E(1) & p_{r(2)2} & E(3) \\
        E(1) & E(2) & p_{r(3)3}
    \end{bmatrix}, \quad 
    \begin{bmatrix}
         E(3) & E(2) & p_{r(1)1} \\
         E(3) & p_{r(2)2} & E(1)\\
        p_{r(3)3} & E(2) & E(1)
    \end{bmatrix}.
\end{equation}

\begin{lemma}
\label{lem:loweralphabetoptimalstructure}
A mechanism with $\supp(P_Y)=\supp(P_X)-k$, that is, with $k$ all-zero columns, can only be \modified{extremal} if it is composed of the elements $p_{ij} \in \{E(j),\Tilde{E}(j)\}, \, \forall j \in [N]$, where $E(j)$ meets the PML constraint with equality for the corresponding column $j$, and $\Tilde{E}(j)$ ensures the mechanism's row-stochasticity. Further, in each column $j$ there are $m_j$ elements $\Tilde{E}(j)$, with
\begin{equation}
    m_j \in \{1,\dots,k+1\}, \quad \sum_{j \in [N]}m_j = N,
\end{equation}
and \modified{each} row contains exactly one of those elements.
\end{lemma}
\begin{proof}
     Without loss of generality, assume a mechanism for which the $k$ last columns are all-zero columns. Consider the first column (which is \modified{strictly positive}). Clearly, from equation \eqref{eq:allNPMLcontradiction}, the maximum number of extremal elements in any non-zero column is upper bounded by $N-1$. The number of overall inequality constraints all columns except the first one can fulfill with equality is therefore upper bounded by
     \begin{equation}
         kN + (N-k-1)(N-1) = N^2 - N - (N - (k + 1)). 
     \end{equation}
     Recall that in order to be an extreme point, any mechanism needs to fulfill $N^2 - N$ inequality constraints with equality. This yields the lower bound on the number of extremal elements in the first column to be $N - (k+1)$. This applies identically to all of the first $N-k$ columns.
     \modified{Hence $m_j \in [k+1]$. Further, the fact that each row can only contain one of the elements $\Tilde{E}(j)$ follows from the same argument as in Lemma \ref{lem:optimalstructure}.}
\end{proof} 

For illustration, consider again the case $N=3$: Assuming that the column meeting $N-2$ constraints is column $j=2$, we get 
\begin{equation}
    \begin{bmatrix}
        E(1) & \Tilde{E}(2) & 0 \\
        E(1) & \Tilde{E}(2) & 0 \\
        \Tilde{E}(1) & E(2) & 0
    \end{bmatrix}, \quad 
    \begin{bmatrix}
         \Tilde{E}(1) & E(2) & 0 \\
         E(1) & \Tilde{E}(2) & 0\\
        E(1) & \Tilde{E}(2) & 0
    \end{bmatrix}
\end{equation}
as two example configurations.

\begin{lemma}
\label{lem:mergeisoptimal}
    \modified{Consider an $N \times N$} extremal mechanism $\boldsymbol{P}$ structured as in Lemma \ref{lem:optimalstructure} and a stochastic mapping $\boldsymbol{W}$, which merges two of $\boldsymbol{P}$'s non-zero columns $\bm P_u$, $\bm P_v$ into one. Then the mechanism $\boldsymbol{W}\circ\boldsymbol{P}$\footnote{\modified{We denote by $\bm W \circ \bm P$ the operation of applying $\bm W$ to the output of $\bm P$.}} is also extremal. 
\end{lemma}
\begin{proof}
    \modified{Fix two columns $u,\,v \in [N]$ and a mapping $\bm W$ merging these two columns into one column corresponding to a new outcome $z$.} Assume \modified{that} the two columns \modified{$u$ and $v$ satisfy} $N-k$ and $N-l$ PML constraints with equality, respectively. \modified{Recall that $\sum_{i}p_{ij}P_X(x_i) = P_Y(y_j)$ and} $(\bm P_j)_i = P_{Y|X=x_i}(y_j) \coloneqq p_{ij}$. Then from the PML constraints we have 
    \begin{equation}
        E(j) \coloneqq \max_{i \in [N]}\,p_{ij} = e^{\epsilon} P_Y(y_j), \, j=u,v 
    \end{equation}
    and therefore 
    \begin{equation}
    \label{eq:summeetsPML}
    \begin{aligned}
         \exp(\ell(X\to z)) &= \frac{E(u) + E(v)}{\sum_i(\bm P_u + \bm P_v)_i P_X(x_i)} \\ &= \frac{E(u) + E(v)}{P_Y(y_u) + P_Y(y_v)} \\&= \frac{e^{\epsilon}P_Y(y_u) + e^{\epsilon}P_Y(y_v)}{P_Y(y_u) + P_Y(y_v)} = e^{\epsilon}.
    \end{aligned}
    \end{equation}
    Further, by Lemma \ref{lem:loweralphabetoptimalstructure} the \modified{indexes} $r(j)$, $r(j')$ of the non-extremal elements are different for any two columns $j \neq j'$. Because of this, there will be a total of $k+l$ of non-extremal elements in the merged column. Now, from the assumption that $\boldsymbol{P}$ is extremal, we know that it fulfills $N^2 - N$ inequality constraints with equality. Since the column merge yields an additional all-zero column fulfilling $N$ non-negativity constraints with equality, the mechanism $(\boldsymbol{W} \circ \boldsymbol{P})$ will meet
    \begin{equation}
    \begin{aligned}
        N^2 - N - &((N-k) + (N-l)) \\&+ N + (N - (k+l)) = N^2 - N
    \end{aligned}
    \end{equation}
    inequality constraints with equality. That is, the mechanism \modified{$(\bm W \circ \bm P)$ is also extremal}.
\end{proof}

\begin{lemma}
\label{lem:mergeexactlyonecolumn}
     Any extremal mechanism \modified{$Q_{Z|X}$ with output support size $|\supp(Z)| = N-k$} can be obtained from an extremal mechanism $P_{Y|X}$ with \modified{output support size $|\supp(P_Y)| = N-k+1$} by merging two of its columns into one. 
\end{lemma}
\begin{proof}
    Suppose the first column of \modified{$Q_{Z|X}$ has $N-m$ extremal elements}, where $1<m \leq k+1$. Then there exists $s,t \geq 1$ such that $m = s+t$. Hence, we can construct $P_{Y|X}$ such that it has a column \modified{with $N-s$ extremal elements}, and an additional column \modified{with $N-t$ extremal elements}, while all other columns are identical to the columns in $Q_{Z|X}$. With this, the mechanism $P_{Y|X}$ has an output \modified{support} size one larger than \modified{$Q_{Z|X}$}. At the same time, Lemma \ref{lem:mergeisoptimal} shows that we can obtain \modified{$Q_{Z|X}$} from the mechanism $P_{Y|X}$ constructed in this way by merging the two newly constructed columns into one. From Lemma \ref{lem:mergeisoptimal} we also know that, if $P_{Y|X}$ is extremal, \modified{$Q_{Z|X}$} is also extremal.
\end{proof}

\begin{lemma}
\label{lem:nooptimalzerocolumns}
    Assume that a mechanism $P_{Y|X}$ satisfies $\epsilon$-PML and has no all-zero column, that is, $P_{Y|X=x_i}(y_j) \neq 0, \, \forall i,j \in [N]$. Assume further that $P_{Y|X}$ is extremal in the sense that it follows the structure presented in Lemma \ref{lem:optimalstructure}. Then any mechanism \modified{$Q_{Z|X}$} that satisfies $\epsilon$-PML and contains an \modified{all-}zero column (i.e., $\exists j \in [N]: Q_{Z|X=x_i}(z_j) = 0, \, \forall i \in [N]$) will not have higher utility given any sub-convex utility function. That is, we have
    \begin{equation}
        U(Q_{Z|X}) \leq U(P_{Y|X}).
    \end{equation}
\end{lemma}
\begin{proof}

    Lemma \ref{lem:mergeexactlyonecolumn} shows inductively that any \modified{extremal mechanism} can be expressed by recursively merging two columns of an extremal mechanism with \modified{output support size $N$} into one while keeping all other columns as they are. In other words, for each extremal mechanism $\boldsymbol{Q}$ with support size $N-k$ there exists an extremal mechanism $\boldsymbol{P}$ with support size $N-k+1$ and a kernel $\boldsymbol{W}$ such that $\boldsymbol{Q} = \boldsymbol{W} \circ \boldsymbol{P}$. Therefore, by the data processing inequality for sub-convex functions \cite[Proposition 17]{extremalmechanismLong}, extremal mechanisms with \modified{output support size smaller than $N$} cannot \modified{achieve} higher utility than extremal mechanisms with \modified{output support size equal to $N$}.
\end{proof}

As previously pointed out, Lemma \ref{lemma:polytope} implies that the optimal solution to the optimization problem \eqref{eq:generalOPTproblem} is one of the extremal mechanisms characterized in the above lemmas. Lemma \ref{lem:nooptimalzerocolumns} then shows that in the high-privacy regime, all extremal mechanisms with $M<N$ can be disregarded. Notice that, given one of the \modified{mechanisms} presented in Lemma \ref{lem:optimalstructure}, the values of $E(j)$ and $\Tilde{E}(j)$ are unique. Therefore, the solution to the maximization problem is unique up to column-permutations of the structured matrices. Since column permutations of a mechanism preserve its utility, \modified{any mechanism satisfying the conditions of Lemma \ref{lem:optimalstructure} is an optimal solution.} Noticing that the mechanism $P^*_{Y|X}$ in Theorem \ref{thrm:highprivacymechanism} has the required structure, and meets $N(N-1)$ inequality constraints with equality, proves that all mechanisms in its equivalence class $[P^*_{Y|X}]$ are optimal in the high-privacy regime.  \qed

\subsection{Proof of Theorem \ref{thrm:allregionseqprioropt}}
\label{app:allregionseqprioropt}

For notational simplicity, denote by $\bm \pi$ and $\bm \rho$ the probability mass functions $P_X$ and $P_Y$, respectively.

Using the homogeneity of sub-linear functions,  we can upper bound any sub-convex utility as
\begin{align}
    U(\bm P) &= \sum_{j=1}^N \mu(\bm P_j) = \sum_{j=1}^N \mu \biggl(\rho_j \bm \lambda_j \biggr) \\ &= \sum_{j=1}^N \rho_j\mu\biggl(\bm \lambda_j \biggr) 
        \leq \max_{j \in [N]}\,\mu\biggl( \bm \lambda_j \biggr), 
\end{align}
        where $(\lambda_{ij}) \eqqcolon \bm \Lambda$ denotes the \textit{lift-matrix}, which we define using the information density $i(x;y)$ as 
        \begin{equation}
            \lambda_{ij} \coloneqq \exp(i(x_i;y_j)) \quad \forall i,j \in [N].
        \end{equation}
        Further, from the PML constraints we have
        \begin{equation}
        \label{eq:postriorbound}
            i(x_i;y_j) \leq \epsilon \quad \forall i,j \in [N]
        \end{equation}
        and since $\mu$ is convex and symmetric, it is Schur-convex \cite{majorization}. Fix some arbitrary $j \in [N]$. Under the given constraints ($\epsilon$-PML, row-stochasticity), $\mu$ will be maximized by the vector $\bm \lambda^*_j$ that majorizes all other vectors $\bm \lambda_j$ satisfying these constraints for some $j \in [N]$. That is, we have $\bm \lambda_j^* \succeq \bm \lambda_j$.
         Since $\epsilon$ is in the $k^{\text{th}}$ privacy region, we know that it can contain no more than $k-1$ zero elements. Further, by \eqref{eq:postriorbound}, the maximum value each of the elements in $\bm \Lambda$ can take is $e^{\epsilon}$. Let $[\bm \lambda_j]$ denote the set of all element permutations of $\bm \lambda_j$ for some fixed $j$. Then we have
        \begin{equation}
            [\bm \lambda^*_j] = [(\underbrace{e^{\epsilon},\dots,e^{\epsilon}}_{{N-k} \text{ times }},r,\underbrace{0,\dots,0}_{k-1 \text{ times }})^T],
        \end{equation}
        where we get the value of $r$ from the constraint
        \begin{equation}
            \sum_{x\in\mathcal X}P_{X|Y=y}(x) = \sum_{i=1}^N \lambda_{ij} \pi_i = \sum_{i = 1}^N\frac{\lambda_{ij}}{N} = 1    
        \end{equation}
        as $r = N - (N-k)e^{\epsilon}$. Note that the value of $r$ is independent of the value of $j$. Hence, we obtain $U(\bm{P}) \leq \mu(\bm \lambda^*_j)$ as an upper bound on the optimal utility. 
        It can be verified that the mechanism given in~\eqref{eq:allregionseqprioropt} attains this bound; thus, it is optimal. \qed


\subsection{Proof of Theorem \ref{thrm:LP}}
\label{app:LP}
Let $\bm \rho$, $\bm \pi$ and $\bm \Lambda$ be defined as in Appendix \ref{app:allregionseqprioropt}. First, consider the following reformulation of problem \eqref{eq:optim_prob_cardinality}:
\begin{subequations}
\begin{align}
\label{subeq:reformulatedOPT}
    \max_{\bm\Lambda, \bm \rho}\, U(\bm \Lambda, \bm \rho) &= \sum_{j=1}^N \mu \biggl(\rho_j \bm \lambda_j \biggr) = \sum_{j=1}^N \rho_j\mu\biggl(\bm \lambda_j \biggr) \\
    \label{subeq:P_yConstraint}
    \text{s.t.} \quad &\sum_{j= 1}^N \rho_j \lambda_{ij} = 1 \quad \forall i \in [N], \\
    \label{subeq:P_xConstraint}
    & \sum_{i=1}^N \pi_i \lambda_{ij} = 1 \quad \forall j \in [N], \\
    \label{subeq:LP-PMLconstraint}
    & 0 \leq \lambda_{ij}  \leq e^{\epsilon} \, \forall i,j \in [N].
\end{align}
\end{subequations}
Note that \eqref{subeq:P_yConstraint} and \eqref{subeq:P_xConstraint} together imply $\sum_y P_Y(y) = 1$. Next, we show that the columns of the optimal lift-matrix $\bm \Lambda^*$ belong to the set $\bm \Lambda^*(k,P_X)$. To see why, let $\bm \rho^*$ denote the optimal distribution in problem \eqref{subeq:reformulatedOPT}. Given this distribution, we find the values of $\bm \lambda^*_j$ by solving the problem
\begin{subequations}
\begin{align}
\label{subeq:reformulatedOPT2}
    \max_{\bm\Lambda}\,  U(\bm \Lambda, \bm \rho^*) &= \sum_{j=1}^N \rho^*_j\mu\biggl(\bm \lambda_j \biggr) \\
    \label{subeq:P_yConstraint2}
    \text{s.t.} \quad &\sum_{j= 1}^N \rho^*_j \lambda_{ij} = 1 \quad \forall i \in [N], \\
    \label{subeq:P_xConstraint2}
    & \sum_{i=1}^N \pi_i \lambda_{ij} = 1 \quad \forall j \in [N], \\
    \label{subeq:LP-PMLconstraint2}
    & 0 \leq \lambda_{ij}  \leq e^{\epsilon} \, \forall i,j \in [N].
\end{align}
\end{subequations}
Since we are maximizing a convex function over a \modified{bounded and convex} polytope, the optimal utility value in this setup is attained by a vertex of \modified{this} polytope. 

\modified{Fix $j$ and $\bm \lambda_j$.} To characterize the vertices, denote by $\Tilde{\bm \pi}(\bm \lambda_j)$ the subset of prior probabilities of all symbols to which $\bm \lambda_j$ assigns a non-zero lift-value, that is,
\begin{equation}
    \Tilde{\bm \pi}(\bm \lambda_j) \coloneqq \{\pi_i \in \bm \pi : \lambda_{ij} > 0\}.
\end{equation}

Then, substituting into \eqref{subeq:P_xConstraint2} and upper bounding with \eqref{subeq:LP-PMLconstraint2} yields the following condition on the probability mass of this set
\begin{equation}
    1 = \sum_{i=1}^N \pi_i \lambda_{ij} = \sum_{i=1}^{|\Tilde{\bm \pi}(\bm \lambda_j)|}\Tilde{\pi}_i( \bm \lambda_j)\lambda_{ij} \leq \sum_{i=1}^{|\Tilde{\bm \pi}(\bm \lambda_j)|}\Tilde{\pi}_i(\bm \lambda_j)e^{\epsilon},
\end{equation}
thus lower bounding the probability of any such subset implied by a vector $\bm \lambda_j$ satisfying $\epsilon$-PML as
\begin{equation}
\label{eq:nonzeroliftbound}
    \sum_{\pi \in \Tilde{\bm \pi(\bm \lambda_j)}}\pi \geq e^{-\epsilon}.
\end{equation}
Define the subset of lift-vectors with $l$ non-zero elements and fulfilling condition \eqref{eq:nonzeroliftbound} as 
\begin{equation}
    \Tilde{\bm \Lambda}(l) \coloneqq \{\bm \lambda_j: \sum_{\pi \in \Tilde{\bm \pi}(\bm \lambda_j)}\pi \geq e^{-\epsilon} \text{ and } |\Tilde{\bm \pi}(\bm \lambda_j)| = l\}.
\end{equation}
Note that, for determining the extremality conditions on these vectors, we can apply the same chain of arguments as \modified{used} in the proof of Theorem \ref{thrm:highprivacymechanism}. That is, there exists an optimal $N\times N$ mechanism $P^*_{Y|X}$ for which all $N$ columns meet exactly $N-1$ inequality constraints with equality. Denote the set of all such vectors in the set $\Tilde{\bm \Lambda}(l)$ as $\Tilde{\bm \Lambda}^*(l)$. Then we get \modified{the set of candidates for the columns of the optimal lift matrix} as $\bm V^*(k,P_X) = \bigcup_{l=0}^{k-1} \Tilde{\bm \Lambda}^*(l)$. 
Due to the convexity of the objective function, and given $\bm \rho^*$ as the optimal distribution on $Y$, it is possible to construct a maximizing solution $\bm \Lambda^*$ using only lift-vectors out of the set $\bm V^*(k,P_X)$. 

Now, all that is left to show is that the optimal distribution on $Y$ can be found by the original optimization problem. Assume the optimal lift-matrix $\bm \Lambda^*$ to be known and let the optimal utility values of column $j$ implied by this solution be denoted by $\mu^*_j \, \forall j \in [N]$. Then the objective function becomes 
\begin{equation}
    \max_{\bm \rho}\, U(\bm \Lambda^*,\bm \rho) = \sum_{j=1}^N \rho_j \mu^*_j,
\end{equation}
which is a linear function of $\bm \rho$. Together with the above derivations, this proves the result. 
\qed
